\documentclass[10pt,a4paper]{article}

\usepackage{amstext,amsgen,latexsym}
\usepackage{amstext,amssymb,amsfonts,latexsym}
\usepackage{theorem}
\usepackage{pifont}
\usepackage[dvips]{graphics,epsfig}

 \newcommand{\bs}{\bigskip}
 \newcommand{\ms}{\medskip}
 \newcommand{\n}{\noindent}
 \newcommand{\s}{\smallskip}
 \newcommand{\hs}[1]{\hspace*{ #1 mm}}
 \newcommand{\vs}[1]{\vspace*{ #1 mm}}

 \newcommand{\real}{\mathbb{R}}

 \newcommand{\nat}{\mathbb{N}}

 \newcommand{\complex}{\mathbb{C}}

 \newcommand{\ie}{\textrm{i.e.},\hspace*{2mm}}
 \newcommand{\eg}{\textrm{e.g.},\hspace*{2mm}}
 \newcommand{\etal}{\textrm{et al.}\hspace*{2mm}}
 \newcommand{\etalc}{\textrm{et al.}}

 \newcommand{\BB}{{\cal B}}

 \newcommand{\FF}{{\cal F}}

 \newcommand{\GG}{{\cal G}}

 \newcommand{\UU}{{\cal U}}

 \newcommand{\fp}{\mathrm{FP}}
 \newcommand{\sharpp}{\#\mathrm{P}}

 \newcommand{\tinymatrices}[4]{\left({\tiny \begin{array}{cc} #1 & #2 \\%
      #3 & #4   \end{array} }\right)}
 
 \newcommand{\tinymatrixthree}[6]{\left({\tiny %
      \begin{array}{ccc} #1 & #2 & #3 \\%
                         #4 & #5 & #6   \end{array} }\right)}

 \newcommand{\tinycomb}[2]{\left({\tiny \begin{array}{c} #1 \\%
      #2 \end{array} }\right)}

 \newcommand{\IFF}{\Longleftrightarrow}
 \newcommand{\implies}{\longrightarrow}

 \def\bbox{\vrule height6pt width6pt depth1pt}

\theoremstyle{plain}
\theoremheaderfont{\bfseries}
 \newtheorem{theorem}{Theorem}[section]
 \newtheorem{lemma}[theorem]{Lemma}
 \newtheorem{proposition}[theorem]{Proposition}

 {\theorembodyfont{\rmfamily}
  }
 {\theorembodyfont{\rmfamily} }
 {\theorembodyfont{\rmfamily} }

 \newtheorem{claim}{Claim}

 \newenvironment{proof}{\par \noindent
            {\bf Proof. \hs{2}}}{\hfill$\Box$ \vspace*{3mm}}

 \newenvironment{proofof}[1]{\vspace*{5mm} \par \noindent
         {\bf Proof of #1.\hs{2}}}{\hfill$\Box$ \vspace*{3mm}}

 \newcommand{\pair}[1]{\langle #1 \rangle}

\newcommand{\ignore}[1]{}

\newcommand{\holant}{\mathrm{Holant}}
\newcommand{\holantstar}{\mathrm{Holant}^{*}}

\newcommand{\sharpcsp}{\#\mathrm{CSP}}
\newcommand{\sharpcspstar}{\#\mathrm{CSP}^{*}}
\newcommand{\APreduces}{\leq_{\mathrm{AP}}}
\newcommand{\APequiv}{\equiv_{\mathrm{AP}}}

\newcommand{\DG}{{\cal DG}}

\newcommand{\DUP}{\mathrm{DUP}}

\begin{document}
\pagestyle{plain}
\setcounter{page}{1}

\begin{center}
{\Large {\bf Approximation Complexity of Complex-Weighted \s\\ 
Degree-Two Counting Constraint Satisfaction Problems}}\footnote{An extended abstract appeared in the Proceedings of the 17th Annual International Computing and Combinatorics Conference (COCOON 2011), Lecture Notes in Computer Science, vol.6842, pp.122--133, Dallas, Texas, USA, August 14-16, 2011.} \bs\\
{\sc Tomoyuki Yamakami}\footnote{Current Affiliation: Department of Information Science, University of Fukui, 3-9-1 Bunkyo, Fukui 910-8507, Japan.} \bs\\
\end{center}


\begin{quote}
{\small 
\n{\bf Abstract:} 
Constraint satisfaction problems have been studied in numerous fields with practical and theoretical interests. In recent years, major breakthroughs  have been made in a study of counting constraint satisfaction problems (or \#CSPs). In particular, a computational complexity classification of bounded-degree \#CSPs has been discovered for all degrees except for two, where  the 
``degree'' of an input instance is the maximal number of times that each input 
variable appears in a given set of constraints. 
Despite the efforts of recent studies, however, a complexity 
classification of  degree-$2$ \#CSPs has eluded from our understandings. 
This paper challenges this open problem and gives its partial solution  
by applying two novel proof techniques---T$_{2}$-constructibility and parametrized 
symmetrization---which are specifically designed to handle 
``arbitrary'' constraints under randomized approximation-preserving  reductions. 
We partition entire constraints into four sets and we  
classify the approximation complexity of all degree-$2$ \#CSPs whose constraints are drawn from two of the four sets  
into two categories: problems computable in polynomial-time or problems that are at least as hard as $\#\mathrm{SAT}$. 
Our proof exploits a close relationship between complex-weighted degree-$2$ 
\#CSPs and Holant problems, which are a natural generalization of complex-weighted \#CSPs.   

\s

\n{\bf Keywords:} constraint satisfaction problem, \#CSP, bounded degree,  AP-reducibility, constructibility, symmetrization, \#SAT, Holant problem, signature
}
\end{quote}

\section{Approximation Complexity of Bounded-Degree \#CSPs}\label{sec:introduction}

{\em Constraint satisfaction problems} (or {\em CSPs}, in short), which are composed of ``variables'' (on appropriate domains) and ``constraints'' among those  variables, have been studied with practical and theoretical interests in various fields, including artificial intelligence, database theory, graph theory, and statistical physics. 
A decision version of CSP asks whether, given a list of constraints over  variables, all the constraints are satisfied simultaneously. 
Schaefer \cite{Sch78} first charted a whole map of the computational complexity of Boolean CSPs (\ie CSPs with constraints on the Boolean domain) according to a fixed list of constraints.  

Of numerous variants of CSPs, in particular, 
a {\em counting CSP} (or {\em \#CSP}) asks 
how many variable assignments satisfy all the given constraints.  
As a typical \#CSP, the {\em counting satisfiability problem} (or \#SAT) 
is to count  the total number of satisfiable assignments for each given logical formula.  This counting problem \#SAT is known to be computationally hard for Valiant's class $\sharpp$ of counting functions \cite{Val79}. 

In the past two decades, a great progress has been observed in a  
study of \#CSPs and their variants. 
The first major leap came in 1996 when Creignou and Hermann \cite{CH96}  discovered a precise classification of all unweighted \#CSPs (\ie \#CSPs with Boolean-valued constraints). 
Their classification theorem asserts that every $\sharpcsp$ whose constraints are all taken from a fixed set $\FF$ (denoted $\sharpcsp(\FF)$) 
can be classified into one of the following 
two categories: polynomial-time computable problems or $\sharpp$-hard problems. This statement is known as a {\em dichotomy theorem} for unweighted \#CSPs.  

In many real-life problems, however, natural constraints often take real or complex values rather than Boolean values. It is therefore quite natural to expand the scope of constraints from Boolean values to real values and beyond.  
An early extension was made by Dyer, Goldberg, and Jerrum \cite{DGJ09} 
to nonnegative rational numbers.  
After a series of vigorous work,  Cai, Lu, and Xia \cite{CLX09x} finally gave a most general form of classification theorem for 
complex-weighted \#CSPs (\ie \#CSPs with complex-valued constraints), provided that arbitrary unary constraints can be freely added to input instances. 
For succinctness, hereafter, we use ``$*$ (star)'' as in ``\#CSP$^*$'' to indicate this extra use of free unary constraints. 

Another major progress has been recently reported in an area of the approximation complexity of \#CSPs. Using a notion of {\em randomized  approximation-preserving reducibility} (or {\em AP-reducibility}, in short) \cite{DGGJ03}, Dyer, Goldberg, and Jerrum \cite{DGJ10} discovered  a complete classification of the approximation complexity of unweighted \#CSPs. Unlike the aforementioned exact complexity case, unweighted \#CSPs are classified into three categories, which include 
an intermediate level between polynomial-time computable problems and $\sharpp$-hard problems. 
This {\em trichotomy theorem} therefore draws a clear contrast between the approximation complexity and the exact complexity of the unweighted \#CSPs.  
Later in 2010, this result was further extended into  complex-weighted 
$\sharpcspstar$s \cite{Yam10a}. 
A recent extensive study has also targeted another important refinement of \#CSPs---{\em bounded-degree \#CSPs}---where the ``degree'' is the maximal number of times that any variable appears in a given set of constraints. 
A complete classification was recently given by Dyer, Goldberg, Jalsenius, and Richerby \cite{DGJR10} to unweighted bounded-degree $\sharpcspstar$s 
when their degree exceeds $2$. Subsequently, 
Yamakami \cite{Yam10b} extended their result to complex-weighted bounded-degree $\sharpcspstar$s. 
We conveniently say that counting problems $A$ and $B$ are 
``AP-equivalent (in complexity)'' when they have the same computational complexity under the aforementioned AP-reductions. 
With a help of this notion,  for any set $\FF$ of constraints,   $\sharpcspstar(\FF)$'s and $\sharpcspstar_{3}(\FF)$'s become AP-equivalent  \cite{Yam10b}, where the subscript ``$3$'' in $\sharpcspstar_{3}(\FF)$ indicates that the maximum degree is at most $3$.   
Nevertheless, degree-$2$ \#CSPs have 
eluded from our understandings and it has remained {\em open} to discover 
a complete classification of the approximation complexity of 
degree-$2$ \#CSPs.  

This paper presents a partial solution to this open problem by exploiting a  fact that the computational complexity 
of \#CSP$^*$s are  closely linked to that 
of {\em Holant problems}, 
where Holant problems were introduced by Cai \etalc~\cite{CLX09x} to 
generalize a framework of \#CSPs (motivated and influenced by Valiant's holographic reductions and algorithms \cite{Val06,Val08}). In this framework, complex-valued  constraints (on the Boolean domain) are simply called {\em  signatures}. A Holant problem then asks to compute the total weights of the products of the values of signatures over all possible edge-assignments to an input graph. Conveniently,  
let $\holantstar(\FF)$ denote a complex-weighted Holant problem whose signatures are either limited to a given set $\FF$ or just unary signatures. 
A close link we exploit here is that $\sharpcspstar_{2}(\FF)$'s and $\holantstar(\FF)$'s are AP-equivalent \cite{Yam10b}, and  
this equivalence makes it possible for us to work on the Holant framework.  

When any permutation of Boolean variables of a signature $f$ does not change the output value of $f$,  the signature $f$ is called 
{\em symmetric}. Typical examples of symmetric signatures include $OR$ (where   $OR(x_1,x_2)$ evaluates the logical formula ``$x_1\vee x_2$'') and $NAND$ (which evaluates ``$not(x_1\wedge x_2)$''). 
All symmetric Holant$^*$ problems (where unary signatures are given for free) were neatly classified  by  Cai, Lu, and Xia \cite{CLX09x}     into two categories: those solvable in polynomial time and those at least 
as hard as the complex-weighted counting satisfiability problem (or $\#\mathrm{SAT}_{\complex}$). 
To obtain this dichotomy theorem, Cai \etal used a technique of Valiant \cite{Val08}, called a {\em holographic transformation}, which transforms signatures without changing solutions of the associated 
Holant$^*$ problems. 

The difference between symmetric signatures and asymmetric ones 
in the case of approximation complexity of \#CSPs with Boolean constraints are quite striking. 
Even for a simple example of binary (\ie arity-$2$) constraints, the symmetric signature $OR$ makes the corresponding counting problem $\sharpcsp(OR)$ $\sharpp$-hard, whereas the asymmetric signature $Implies$ (where $Implies(x_1,x_2)$ evaluates the propositional formula ``$x_1\supset x_2$'') makes $\sharpcsp(Implies)$ sit between the set of polynomial-time solvable problems and the set of $\sharpp$-hard problems \cite{DGJ10} and $\sharpcsp(Implies)$ has been speculated to be intractable.   

In this paper, we give two approximation classification theorems for complex-weighted degree-$2$  $\sharpcspstar$s.  Our major contributions are two fold: (1) we present  a
systematic technique of handling arbitrary signatures and (2) we demonstrate two classification theorems for approximation complexity of complex-weighted  $\sharpcspstar$s associated with particular sets of  signatures. 
To be more precise, in the first classification theorem (Theorem \ref{outside-SIG}), 
we first define a ternary signature set $SIG$ and prove that, for any signature $f$ outside of $SIG$, $\sharpcspstar_{2}(f)$ 
is at least as hard as $\#\mathrm{SAT}_{\complex}$ 
(\ie a complex-valued version of $\#\mathrm{SAT}$). 
This result leaves the remaining task of focusing on 
ternary signatures residing within $SIG$. For our convenience, we will  split $SIG$ 
into three parts---$SIG_0$, $SIG_1$, and $SIG_2$---and, in the second classification theorem, when  all signatures are drawn from $SIG_1$, 
we provide with a complete classification of all degree-$2$ $\sharpcspstar$s. The other two sets will be handled in separate papers due to 
their lengthy proofs.  
The second classification theorem (Theorem \ref{main-theorem}) is roughly stated as follows: for any set $\FF$ of signatures in $SIG_1$, if $\FF$ is included in a particular signature set, called $\DUP$,  
then $\sharpcspstar_{2}(\FF)$ is solvable in polynomial time; otherwise, $\sharpcspstar_{2}(\FF)$ is computationally hard for 
$\sharpp_{\complex}$ under AP-reductions, where $\sharpp_{\complex}$ is a complex-valued version of $\sharpp$ (see, \eg \cite{Yam10a}).  In fact, we can precisely describe the requirements for asymmetric signatures to be $\sharpp_{\complex}$-hard. 
Proving these two theorems require novel ideas and new technical tools: {\em T$_{2}$-constructibility} and {\em parameterized  symmetrization scheme} of asymmetric signatures. 

Our proofs of the aforementioned main theorems proceed in the following way. 
{}From an arbitrary ternary signature $f$, we nicely construct a new ``ternary''  signature, denoted $Sym(f)$,  so that $Sym(f)$ becomes symmetric. 
This process, which is a form of (simple) symmetrization scheme, is carried out by  T$_{2}$-construction, and this construction ensures that the corresponding problem $\sharpcspstar_{2}(f)$ is AP-equivalent to 
$\sharpcspstar_{2}(Sym(f))$. 
When $f$ is outside of $SIG$, $\sharpcspstar_{2}(Sym(f))$ 
further becomes AP-equivalent to certain symmetric Holant$^*$ problems, and thus we can appeal to the dichotomy theorem of Cai \etal for 
symmetric Holant$^*$ problems. When $f$ is in $SIG_1$, on the contrary, we need another symmetric ``binary'' signature alongside $Sym(f)$. Employing another symmetrization scheme, we T$_{2}$-construct such a signature, denoted $SymL(f)$, from $f$. Moreover, this new signature is ``parametrized'' so that we can discuss an infinite number of similar signatures simultaneously. To apply  Cai \etalc's dichotomy theorem, the two symmetrized signatures must fail to meet a few special conditions.    
To prove that this is indeed the case, we falsely assume that those conditions are met. Now, we  translate the conditions into a set of certain low-degree multivariate polynomial equations that have a common solution in $\complex$. We then try to argue that there is no such common solution, contradicting our initial assumption.  Notably, this argument requires only an elementary analysis of low-degree polynomial equations and the whole analysis is easy and straightforward to follow. This nice feature is an advantage and strength of our argument. 

\ms

To prove the two main theorems, the rest of this paper is organized as follows. First, we describe fundamental notions and notations in Section \ref{sec:preliminaries}, including signatures, Holant problems, \#CSPs, AP-reduction, and holographic transformation. 
We then introduce two new technical tools---T$_{2}$-constructibility and parametrized symmetrization---for the description of the 
proofs of our main theorems (Theorems 
 \ref{outside-SIG}--\ref{main-theorem}). 
The notion of T$_{2}$-constructibility is explained in Section \ref{sec:T-const}, and the notions of (simple) symmetrization scheme and parametrized symmetrization scheme appear respectively in Sections \ref{sec:asymmetric-main} and 
\ref{sec:first-symmetrization}. 
Many fundamental properties of those symmetrization schemes are presented in Section \ref{sec:SymL(f)}. 
Theorem \ref{outside-SIG} relies on Proposition \ref{SAT-ONE3-EQ2} and its proof  appears in Section \ref{sec:SAT-reduction}. 
In contrast, the proof of Theorem \ref{main-theorem} uses two key propositions, Propositions \ref{computability-result}--\ref{h0+h2=h1+h3=0}, where  Proposition \ref{computability-result} is proven in Section \ref{sec:key-props}, and 
the proof of Proposition \ref{h0+h2=h1+h3=0} is given in 
Section \ref{sec:outline-proof} based on Proposition \ref{SymL-non-degenerate}--\ref{prop-type-III}. Finally, Proposition \ref{SymL-non-degenerate} is proven  in Section \ref{sec:proof-SymL}, and Proposition \ref{prop-type-I}--\ref{prop-type-III} are explained in Sections \ref{sec:type-I}--\ref{sec:type-III}, completing the proof of Proposition \ref{h0+h2=h1+h3=0}.

\section{Fundamental Notions and Notations}\label{sec:preliminaries}

We briefly present fundamental notions and notations, which will be used in later sections. Let $\nat$ denote the set of all {\em natural numbers} (\ie non-negative integers). For convenience, the notation $\nat^{+}$ expresses $\nat-\{0\}$.  Moreover, $\real$ and $\complex$ denote respectively the sets of all {\em real numbers} and of all {\em complex numbers}. 
For any complex number $\alpha$, 
$|\alpha|$ and $\arg(\alpha)$ denote the {\em absolute value} and the {\em argument} of $\alpha$, respectively. 
For each number $n\in\nat^{+}$, $[n]$ denotes the integer set $\{1,2,\ldots,n\}$. 
For a position integer $k$, let $S_k$ denote the set of all permutations over $[k]$. For brevity, we express each permutation $\sigma\in S_k$ as $(a_1a_2\ldots a_k)$  to mean that $\sigma(i)=a_{i}$ for every index $i\in[k]$. 
We always treat vectors as {\em row vectors}, unless stated otherwise. 
To simplify descriptions of compound conditions and requirements among Boolean variables, we informally use {\em logical connectives}, 
such as ``$\wedge$'' (AND), ``$\vee$'' (OR), and ``$not$'' (NOT). An  example of such usage is: $(g_1=0\wedge g_0+g_2=0)\vee not(g_0=g_2=0)$.

\subsection{Signatures and Relations}

The most fundamental concept in this paper is ``signature'' on the 
Boolean domain. Instead of the conventional term ``constraint,'' we intend in this paper to use this term ``signature.'' 
A {\em  signature}  of arity $k$ is a 
complex-valued function of arity $k$; that is, $f$ is a map from $\{0,1\}^{k}$ to $\complex$. 
Assuming the standard lexicographic order on $\{0,1\}^{k}$, 
we conveniently express $f$ as a row-vector consisting of its output values, which can be identified with an element in the space $\complex^{2^{k}}$. 
For instance, if $f$ has arity $2$, then $f$ is expressed as $(f(00),f(01),f(10),f(11))$. 
A signature $f$ is called {\em symmetric} if $f$'s values depend only on the {\em Hamming weight} of inputs.  
An {\em asymmetric} signature, on the contrary, is a 
signature that is not symmetric. 
When $f$ is an arity-$k$ symmetric function, we use another succinct notation  $f=[f_0,f_1,\ldots,f_k]$, where each $f_i$ is the value of $f$ on inputs of 
Hamming weight $i$. For example, the equality function ($EQ_k$) of arity $k$ is expressed as $[1,0,\ldots,0,1]$ ($k-1$ zeros).  
{\em Unary signatures} (\ie signatures of arity $1$), in particular, play an essential role in this paper. 

A {\em relation} of arity $k$ is a subset of $\{0,1\}^k$. Such a relation can be also viewed as a function mapping Boolean variables to $\{0,1\}$ (\ie $x\in R$ iff $R(x)=1$, for every $x\in\{0,1\}^k$) and it can be treated as a ``Boolean'' signature. For instance, logical relations $OR$, $NAND$,  
and $Implies$ are expressed as ``signatures'' in the following obvious manner: $OR=[0,1,1]$, $NAND=[1,1,0]$,  and $Implies=(1,1,0,1)$. 
In addition, we define $ONE_3=[1,1,0,0]$, which means that the total number of $1$s in any satisfying assignment should equal one. 

To simplify our further descriptions, it is useful to introduce the following two special sets of signatures. 
First, let $\UU$ denote the set of all unary signatures. 
Next, let $\DG$ denote the set of all signatures $f$ of arity $k$ that are expressed by products of $k$ unary functions, which are applied respectively to $k$ variables. A signature in $\DG$ is called {\em degenerate}.  Note that, for ternary symmetric signature $f=[a_0,a_1,\ldots,a_k]$, $f$ is non-degenerate if and only if the rank of $\tinycomb{a_0\;a_1\;\cdots\; a_{k-1}}{a_1\;a_2\;\cdots\;a_k\;\;\;\;}$ is exactly two (see, \eg \cite{CLX09x}). 

\subsection{\#CSPs and Holant Problems}

In an undirected bipartite graph $G=(V_1|V_2,E)$ (where $V_1,V_2$ are {\em vertex sets} and $E$ is an {\em edge set}), 
all nodes in $V_1$ appear on the left-hand side and all nodes in $V_2$ appear on the right-hand side of the graph. 
For any  vertex $v$, the {\em incident set} $E(v)$ of $v$ is a set of all edges {\em incident} on $v$, and $deg(v)$ is the {\em degree} of $v$. For any matrix $A$, the notation $A^{T}$ denotes the {\em transposed matrix} of $A$. 

Let us define complex-weighted (Boolean) $\sharpcsp$ problems.   Throughout this paper, the notation $\FF$ often denotes an arbitrary set of signatures of arity at least $1$.  
Conventionally, the term ``constraint'' is used to describe 
a function mapping variables on a certain domain; nonetheless, 
as we have stated in the previous subsection, 
we wish to use the term ``signature'' instead. 
Limited to a given set $\FF$, a 
complex-weighted $\sharpcsp$ problem, denoted $\sharpcsp(\FF)$, takes as an input instance a finite subset $H$ of all elements of the form $\pair{h,(x_{i_1},x_{i_2},\ldots,x_{i_k})}$, where a signature $h\in\FF$ is defined on   $(x_{i_1},x_{i_2},\ldots,x_{i_k})$ of Boolean variables $\{x_1,x_2,\ldots,x_n\}$ with $i_1,\ldots,i_k\in[n]$,  and the problem outputs the complex value: 
\[
\sum_{x_1,x_2,\ldots,x_n\in\{0,1\}} \prod_{\pair{h,x'}\in H} h(x_{i_1},x_{i_2},\ldots,x_{i_k}),
\] 
\sloppy where $x'= (x_{i_1},x_{i_2},\ldots,x_{i_k})$. For brevity, we often express  $h(x_{i_1},x_{i_2},\ldots,x_{i_k})$ to mean  $\pair{h,(x_{i_1},x_{i_2},\ldots,x_{i_k})}$ whenever it is clear from the context. 
The {\em degree} of an input instance to $\sharpcsp(\FF)$ 
is the greatest number of times that any variable 
appears among its signatures.   
For any positive integer $d$, $\sharpcsp_{d}(\FF)$ expresses the restriction of $\sharpcsp(\FF)$ to instances of degrees at most $d$.

We can view a counting problem \#CSPs from a slightly different perspective, 
known as a {\em Holant framework}, and we pay our attention to so-called  {\em Holant problems}.  
An input instance to a Holant problem is a signature grid that  
contains an undirected  graph $G$, in which all nodes are labeled by signatures in $\FF$.   
More formally, following the terminology developed in \cite{CL07,CL08}, 
 we define 
a {\em bipartite Holant problem} $\holant(\FF_1|\FF_2)$  as a counting problem that takes a {\em (bipartite) signature grid}  $\Omega =(G,\FF'_1|\FF'_2,\pi)$, where $G=(V_1|V_2,E)$ is a finite undirected bipartite graph, two ``finite'' subsets $\FF'_1\subseteq \FF_1$ and $\FF'_2\subseteq \FF_2$, and a {\em labeling function} $\pi:V_1\cup V_2\rightarrow\FF_1'\cup\FF'_2$ such that  $\pi(V_1)\subseteq \FF'_1$ and $\pi(V_2)\subseteq \FF'_2$, and each vertex $v\in V_1\cup V_2$ is labeled by a signature $\pi(v): \{0,1\}^{deg(v)}\rightarrow\complex$. For convenience, we often write $f_{v}$ for $\pi(v)$. 
Let $Asn(E)$ be the set of all edge assignments $\sigma: E\rightarrow \{0,1\}$. The objective of this problem is to compute the following value $\holant_{\Omega}$: 
\[
\holant_{\Omega} = \sum_{\sigma\in Asn(E)} \prod_{v\in V}f_{v}(\sigma|E(v)), 
\]    
where $\sigma|E(v)$ denotes the binary string $(\sigma(w_1),\sigma(w_2),\cdots,\sigma(w_k))$ if $E(v)=\{w_1,w_2,\ldots,w_k\}$, sorted in a certain pre-fixed order by $f$. 

We often view 
$\sharpcsp(\FF)$ (as well as $\sharpcsp_{d}(\FF)$) as a special case of bipartite Holant problem of the following form: an instance to  $\sharpcsp(\FF)$ is a bipartite graph $G$, where all vertices on the left-hand side, each of which represents a variable, are labeled by equality functions ($EQ_k$) and all vertices on the right-hand side are labeled by constraints. Whenever variables appear in constraints, edges are drawn between their corresponding nodes on each side of the graph. 
In terms of Holant problems, therefore,  $\sharpcsp(\FF)$ coincides with 
$\holant(\{EQ_{k}\}_{k\geq1}|\FF)$. Throughout this paper, we interchangeably take these two different views of 
complex-weighted $\sharpcsp$ problems. With this Holant viewpoint, the degree of an instance is just the maximum degree of nodes that appear on the left-hand side of a bipartite graph in the instance.

The following abbreviations are useful in this paper;   
we write $\sharpcsp(f,\FF,\GG)$ to mean $\sharpcsp(\{f\}\cup \FF\cup\GG)$ and $\holant(f,\FF_1|\FF_2,\GG)$ to 
mean $\holant(\{f\}\cup\FF_1|\FF_2\cup\GG)$, for example.  
In particular, we abbreviate $\sharpcsp(\UU,\FF)$,  $\sharpcsp_{d}(\UU,\FF)$,  and $\holant(\UU,\FF_1|\UU,\FF_2)$ as  $\sharpcspstar(\FF)$, $\sharpcspstar_{d}(\FF)$, and $\holantstar(\FF_1|\FF_2)$, respectively. 

In the end, as a concrete example of counting problem, we introduce a complex-weighted version of the {\em counting satisfiability problem}, denoted $\#\mathrm{SAT}_{\complex}$ in \cite{Yam10a}.   
Let $\phi$ be any propositional formula and let $V(\phi)$ denote the set of all variables that appear in $\phi$. For this formula $\phi$, we consider 
a series $\{w_x\}_{x\in V(\phi)}$ of {\em node-weight functions} $w_x:\{0,1\}\rightarrow \complex-\{0\}$. Given the pair $(\phi,\{w_x\}_{x\in V(\phi)})$, $\#\mathrm{SAT}_{\complex}$ asks to compute the sum of all weights $w(\sigma)$ for every truth assignment $\sigma$ that satisfies $\phi$, where $w(\sigma)$ is the product of all $w_x(\sigma(x))$ for any $x\in V(\phi)$.

\subsection{FP$_{\complex}$ and AP-Reducibility}\label{sec:randomized-scheme}

To compare the exact complexities of two Holant problems, 
Cai \etalc~\cite{CLX09x} utilized a complex-valued analogue of 
(polynomial-time) Turing reducibility. In contrast, for approximation complexity, 
Dyer, Goldberg, Greenhill, and Jerrum \cite{DGGJ03} introduced 
so-called ``AP-reducibility'' to measure the approximation complexity 
of  various unweighted \#CSPs. Here, we adapt their notion of AP-reducibility. Since all \#CSP$^*$s can be treated as 
complex-valued functions mapping from 
$\{0,1\}^*$ to $\complex$, it suffices for us to develop necessary methodology concerning only complex-valued functions. 

The following notational conventions are taken from  \cite{Yam10a,Yam10b}. 
The notation $\fp_{\complex}$ denotes the collection of all string-based     
functions $f:\{0,1\}^*\rightarrow\complex$ that can be computed deterministically in time polynomial in the lengths of inputs. 
A {\em randomized approximation scheme} for (complex-valued) $F$ is a randomized algorithm that takes a standard input $x\in\Sigma^*$ together with an {\em error tolerance parameter} $\varepsilon\in(0,1)$, and outputs values $w$ with probability at least $3/4$ for which  
\[
2^{-\epsilon} \leq \left| \frac{w}{F(x)}\right| \leq 2^{\epsilon} 
\hs{3} \text{and} \hs{3} 
\left| \arg\left( \frac{w}{F(x)}\right) \right|\leq 2^{\epsilon},
\]
where  we conventionally assume that, whenever $|F(x)|=0$ or $\arg(F(x))=0$, we instead require $|w|=0$ or $|\arg(w)|\leq 2^{\epsilon}$, respectively. 
Furthermore, when 
a randomized approximation scheme for $F$ runs in time polynomial in $(|x|,1/\varepsilon)$, we call it a 
{\em fully polynomial(-time) randomized approximation scheme} 
(or simply, {\em FPRAS}) for $F$. 

Now, we are ready to introduce the desired reduction 
between complex-valued functions in our approximation context. 
Given two functions $F$ and $G$, a {\em polynomial-time randomized  approximation-preserving reduction} (or {\em AP-reduction}) from $F$ to $G$ is a randomized algorithm $M$ that takes a pair $(x,\varepsilon)\in\Sigma^*\times(0,1)$ as input instance,  
uses an arbitrary randomized approximation scheme $N$ for $G$ as oracle, 
and satisfies the following three conditions:
(i) $M$ is still a randomized approximation scheme for $F$ independent of a choice of $N$ for $G$;  
(ii) every oracle call made by $M$ is of the form $(w,\delta)$ in $\Sigma^*\times(0,1)$ with $1/\delta \leq p(|x|,1/\varepsilon)$, where $p$ is a fixed polynomial, and its answer is the outcome of $N$ on $(w,\delta)$; and (iii) the running time of $M$ is upper-bounded 
by a certain polynomial in $(|x|,1/\varepsilon)$, which is not depending on the choice of $N$ for $G$. If such an AP-reduction exists, then we 
say that $F$ is {\em AP-reducible} to $G$ and we write $F\APreduces G$. 
If $F\APreduces G$ and $G\APreduces F$, then $F$ and $G$ are 
said to be {\em AP-equivalent} and we use the notation $F\APequiv G$.  
 
The following basic properties of AP-reductions are straightforward from the definition of $\sharpcspstar_{2}(\FF)$'s: given  two signature sets  $\FF$ and $\GG$,  if $\FF\subseteq \GG$, then 
$\sharpcspstar_{2}(\FF)\APreduces \sharpcspstar_{2}(\GG)$. 

Lemma \ref{holant-reduction} gives additional useful properties. 
To prove the lemma, we need the following results proven in \cite{Yam10b}: 
for any signature set $\FF$,   
 $\sharpcspstar(\FF)\APequiv \sharpcspstar_{3}(\FF) \APequiv \holantstar(EQ_3|\FF)$ and  
$\sharpcspstar_{2}(\FF) \APequiv \holantstar(EQ_2|\FF)$. 

\begin{lemma}\label{holant-reduction}
(1) For any signature $f$, $\holantstar(EQ_2|f) \APreduces \holantstar(EQ_3|f)$.  
(2) For any set $\FF$ of signatures,  $\holantstar(EQ_2|\FF)\APreduces \sharpcspstar(\FF)$.
\end{lemma}

\begin{proof}
(1) This can be easily shown by replacing,  with $\sum_{x_3\in\{0,1\}} EQ_{3}(x_1,x_2,x_3)\cdot [1,1](x_3)$, each signature $EQ_2(x_1,x_2)$ that appears in any signature grid to $\holantstar(EQ_2|\FF)$.  

(2) Using (1), we obtain $\holantstar(EQ_2|\FF) \APreduces 
\holantstar(EQ_3|\FF)$. 
The remaining AP-equivalence   $\holantstar(EQ_3|\FF) \APequiv  \sharpcspstar(\FF)$  follows from \cite{Yam10b}. 
\end{proof}

\subsection{Holographic Transformation} 

The notion of {\em holographic transformation} was introduced by Valiant \cite{Val02b,Val08} to extend the scope of the application of holographic algorithms. Cal and Lu \cite{CL08} later contributed to its abstract   formulation.  
Holographic transformation is one of the few technical tools that still work together with AP-reducibility. 
Since each signature $f$ is expressed as a row vector, whenever we want to use a column-vector form of $f$, 
we formally write $f^{T}$ to avoid any confusion that may incur. 

We fix a $2\times2$ nonsingular matrix $M$ and let $f$ and $g$ be signatures of arity $k$ and $m$, respectively. For any signature grid $\Omega=(G,\{g\}|\{f\},\pi)$, we define another signature grid $\Omega'$ by simply replacing the nodes's labels $g$ and $f$ respectively with $f(M^T)^{\otimes k}$ and $g(M^{-1})^{\otimes m}$, where $\otimes$ means the {\em tensor product}. A key observation made by Valiant is that $\holant_{\Omega}$ equals $\holant_{\Omega'}$. More generally, let $\FF$ and $\GG$ be any two sets of signatures. We conveniently write ${\GG}(M^{-1})^{\otimes}$ for the set $\{ g(M^{-1})^{\otimes k} \mid   f\in\GG, \text{$f$ has arity $k$}\}$ and $\FF (M^T)^{\otimes}$ for the set $\{ f(M^T)^{\otimes k}\mid f\in\FF, \text{$f$ has arity $k$}\}$.
(Note that, for any vectors $f,g$ of dimension $k$, the equation $h=f(M^T)^{\otimes k}$ is equivalent to the equation 
$h^T = M^{\otimes k}f^T$.)  By the above observation, 
holographic transformation obviously preserves the exact complexity of Holant problems under Turing reductions, and thus obtain 
Valiant's so-called {\em Holant theorem}: $\holant(\GG|\FF)$ is Turing equivalent to $\holant({\GG}(M^{-1})^{\otimes}|{\FF}(M^{T})^{\otimes})$ 
for any $2\times 2$  nonsingular complex matrix $M$ (see, \eg \cite{CL08,CL07,CLX09x} 
for a discussion). It is important to note 
that the Holant theorem is still valid under AP-reductions, because we can trivially construct an AP-reduction machine computing, \eg $\holant_{\Omega'}$ from $\holant_{\Omega}$ defined above. Since unary signatures are transformed into unary signatures, we therefore obtain the following statement.  

\begin{lemma}\label{holographic-transform}
$\holantstar(\GG|\FF)\APequiv \holantstar({\GG}(M^{-1})^{\otimes}|{\FF}(M^{T})^{\otimes})$ for any $2\times 2$  nonsingular complex matrix $M$.  
\end{lemma}

This lemma will be extensively used to prove one of the four key propositions, namely, Proposition \ref{SAT-ONE3-EQ2}.

\section{Main Theorems}\label{sec:main-theorem}

Now, we challenge an unsolved question of determining the 
approximation complexity of degree-$2$ 
$\sharpcspstar$s. With a great help of two new powerful techniques for 
``arbitrary'' signatures, we can give a partial answer to this question by presenting two main theorems---Theorems \ref{outside-SIG} and \ref{main-theorem}---for the degree-$2$ $\sharpcspstar$s with ternary signatures. The first technical tool is a modification of {\em T-constructibility}, which was shown effective for unbounded-degree $\sharpcspstar$s \cite{Yam10a}. The second tool is a clear, systematic method of transforming arbitrary signatures into slightly more complicated but ``symmetric'' signatures. These techniques will be explained in details in the subsequent sections. The two theorems may suggest a future direction of the intensive research on \#CSPs (on an arbitrary domain). 

\subsection{Symmetric Signatures of Arity 3}

To state our main theorems, we begin with a short discussion on symmetric signatures of arity $3$. 
Recently, a crucial progress was made by Cai, Lu, and Xia \cite{CLX09x} in the field of Holant problems, in particular, ``symmetric'' Holant$^{*}$  problems. 
A counting problem $\holantstar(f)$ with a symmetric signature $f$ is shown to be classified into only two types: either it is polynomial-time solvable or it is at least as hard as $\#\mathrm{SAT}_{\complex}$. 
In this classification, 
Cai \etal recognized two useful categories of ternary symmetric signatures. A ternary signature of the first category has the form $[a,b,-a,-b]$ 
with two constants $a,b\in\complex$. In contrast, a ternary signature  $[a,b,c,d]$ of the second category satisfies the following technical condition: there exist two constants $\alpha,\beta\in\complex$ (not both zero) for which 
$\alpha a+\beta b-\alpha c=0$ and $\alpha b+\beta c-\alpha d=0$. For later convenience, we call this pair $(\alpha,\beta)$ the {\em binding coefficients} of the signature. 
To simplify our description, the notations $Sig^{(1)}$ and $Sig^{(2)}$  
respectively  
denote the sets of all signatures of the first category and of the second category.   

\sloppy Regarding $Sig^{(1)}$ and $Sig^{(2)}$, Cai \etal proved three key lemmas, which lead to their final dichotomy theorem for symmetric Holant$^*$ problems: unless target Holant$^*$ problems are in $\fp_{\complex}$, 
they are Turing reducible to one of the following three problems,  $\holantstar(EQ_3|OR)$, $\holantstar(EQ_3|NAND)$, and 
$\holantstar(ONE_3|EQ_2)$. For later convenience, we define $\BB =\{(EQ_3|OR), (EQ_3|NAND), (ONE_3|EQ_2)\}$. 
Notice that the proofs of their lemmas require only a holographic transformation technique and a ``realizability'' technique. Since these tools still work in our  approximation context, we obtain the following three statements, which become a preparation to the description of our main theorems.

\begin{lemma}\label{cai-lemmas}
Let $f$ be any ternary non-degenerate symmetric signature and let $g=[c_0,c_1,c_2]$ be any non-degenerate signature. Each of the following statements holds. 
\begin{enumerate}\vs{-1}
\item\label{cai-third-case}  
If $f\not\in Sig^{(1)}\cup Sig^{(2)}$, then 
there exists a pair $(g_1|g_2)\in\BB$ such that  $\holantstar(g_1|g_2)\APreduces \holantstar(EQ_2|f)$. 
\vs{-2}
\item\label{cai-first-case} If $f\in Sig^{(1)}$, 
$g\not\in\{[\lambda,0,\lambda]\mid \lambda\in\complex\}$, 
and $c_0+c_2\neq 0$, then there exists  a pair $(g_1|g_2)\in\BB$ such that  $\holantstar(g_1|g_2)\APreduces \holantstar(EQ_2|f,g)$. 
\vs{-2}
\item\label{cai-second-case} If $f\in Sig^{(2)}$  
with its binding coefficients $(\alpha,\beta)$,  
 $g\not\in\{[2\alpha\lambda,\beta\lambda,2\alpha\lambda]\mid \lambda\in\complex\}$, and $\alpha c_0+ \beta c_1 - \alpha c_2\neq 0$, 
then there exists  a pair $(g_1|g_2)\in\BB$ such that  $\holantstar(g_1|g_2)\APreduces \holantstar(EQ_2|f,g)$. 
\end{enumerate} 
\end{lemma}

\begin{proof}
Here, we will prove only (2). In this proof, we need a notion of T$_{2}$-constructibility as well as Lemma \ref{T-const-reduction}, which will be described in Section \ref{sec:T-const-tech}.   
Following an argument of Cai, Lu, and Xia \cite{CLX09x}, 
for given signatures $f$ and $g$, 
we first choose a pair $(g_1|g_2)\in\BB$, a signature $h$, and  a $2\times2$ nonsingular matrix $M$ such that  $EQ_2=g_2(M^{-1})^{\otimes 2}$ and  $h=g_1(M^T)^{\otimes 3}$; in other words,    
$\holantstar(g_2|g_1)$ is transformed into $\holantstar(EQ_2|h)$ 
by Valiant's holographic transformation.  Notice that $\holantstar(g_1|g_2)$ and $\holantstar(g_2|g_1)$ are essentially identical.  
By Lemma \ref{holographic-transform}, we conclude that $\holantstar(g_1|g_2)\APreduces \holantstar(EQ_2|h)$. By analyzing the argument in \cite{CLX09x}, we can show that, with a certain finite subset $\FF\subseteq\UU$,  
$h$ is T$_{2}$-constructed from signatures in $\FF\cup \{f,g\}$. 
Therefore, by applying Lemma \ref{T-const-reduction}, we immediately obtain the desired AP-reduction: $\holantstar(g_1|g_2)\APreduces \holantstar(EQ_2|f,g)$. 
\end{proof}

As discussed earlier, Holant$^*$ problems $\holantstar(g_1|g_2)$ with  $(g_1|g_2)\in\BB$ are at least as hard as $\#\mathrm{SAT}_{\complex}$ under Turing reductions \cite{CLX09x}.  
When dealing with complex numbers, in general, 
it is not immediately clear that Turing reductions can be automatically replaced by AP-reductions, because a number of ``adaptive'' queries made by Turing reductions might possibly violate certain requirements imposed on the definition of AP-reduction. Despite such a concern, we will be able to prove in Proposition \ref{SAT-ONE3-EQ2} that those problems are indeed 
AP-reduced from $\#\mathrm{SAT}_{\complex}$, and thus Lemma \ref{cai-lemmas} is  still applicable to obtain the $\sharpp_{\complex}$-hardness of certain $\sharpcspstar_{2}(\FF)$'s.  

\subsection{Arbitrary Signatures of Arity 3}\label{sec:asymmetric-main}

Finally, we turn our attention to arbitrary signatures of arity $3$ 
and their associated degree-$2$ $\sharpcspstar$s. 
We have already seen the dichotomy theorem of Cai \etalc~\cite{CLX09x} 
for symmetric Holant$^*$ problems hinge on two particular signature sets  $Sig^{(1)}$ and $Sig^{(2)}$.  In order to obtain  
a similar classification theorem for 
{\em all} ternary signatures, we wish to take the first systematic approach by introducing two useful tools. Since these tools are not limited to a particular type of signatures, as a result, we will obtain a general classification of the approximation complexity of degree-$2$ $\sharpcspstar$s. 
The first new technical tool is  ``symmetrization'' of 
arbitrary signatures. 
Another new technical tool is ``constructibility'' that bridges between symmetrization and degree-$2$ $\sharpcspstar$s. 
Throughout this section, let $f$ denote any ternary signature with complex components; in particular, we assume that $f=(a,b,c,d,x,y,z,w)$.  
Here, we introduce a simple form of {\em symmetrization} of $f$, denoted $Sym(f)$, as follows:  
\begin{equation}\label{Sym(f)-def}
Sym(f)(x_1,y_1,z_1)
= \sum_{x_2,y_2,z_2\in\{0,1\}} f(x_1,x_2,z_2)f(y_1,y_2,x_2)f(z_1,z_2,y_2).
\end{equation}
This symmetrization $Sym(f)$ plays a key role in the description of 
our main theorems. 
As its name suggests, the symmetrization transforms any signature into a symmetric signature.   

\begin{lemma}
For any ternary signature $f$, $Sym(f)$ is a symmetric signature.
\end{lemma}

\begin{proof}
Let $x_1,y_1,z_1$ be any three variables. First, we want to show that the value  $Sym(f)(x_1,y_1,z_1)$ coincides with $Sym(f)(y_1,z_1,x_1)$.  Let us focus on $Sym(f)(x_1,y_1,z_1)$, which is calculated according to Eq.(\ref{Sym(f)-def}).  To terms inside the summation of Eq.(\ref{Sym(f)-def}), we  apply the following map: $x_2\mapsto z_2$, $z_2\mapsto  y_2$, and $y_2\mapsto x_2$. Although this map does not change the actual value of  $Sym(f)(x_1,y_1,z_1)$, exchanging the order of three $f(\cdot)$'s 
inside the summation immediately produces the valid definition of $Sym(f)(y_1,z_1,x_1)$. Thus, $Sym(f)(x_1,y_1,z_1)$ equals $Sym(f)(y_1,z_1,x_1)$.   Similarly, we can handle the other remaining cases.  Since the signature $Sym(f)$ is independent of the input-variable order, it should be symmetric. 
\end{proof}

Although most of the fundamental properties will be provided in Section \ref{sec:property-h}, here we present a significant nature of the symmetrization: $Sym(\cdot)$ behaves quite differently on $Sig^{(1)}$ and $Sig^{(2)}$. 

\begin{lemma}\label{Sig1-Sym(f)}
Let $f=(a,b,c,d,x,y,z,w)$ be any ternary symmetric signature. 
(1) If $f\in Sig^{(1)}$, then $Sym(f)$ 
is in $\DG$. (2) Assume that $f\in Sig^{(2)}$ with binding coefficients  $(\alpha,\beta)$. If either $\alpha\beta=0$ or $\alpha\beta\neq0\wedge \left(\beta/\alpha+a/b\right)^2=-1$, then $Sym(f)$ is in $Sig^{(2)}$. 
\end{lemma}

\begin{proof}
Let us consider any ternary symmetric signature $f=(a,b,c,d,x,y,z,w)$. 
When $f\in Sig^{(1)}$, $f$ can be expressed as $[a,b,-a,-b]$. Hence, it follows that (1') $a+d=x+w=0$ and (2') $a^2+bc = bc+d^2 = a^2+b^2$. Using  these equations, the value $h_1$ described in Eq.(\ref{h-123-one}) can be simplified to $(a^2+b^2)x + (a^2+b^2)w$, which obviously equals $0$.  Similarly, with a help of (1')--(2'), Eq.(\ref{h-123-zero})\&(\ref{h-123-two})--(\ref{h-123-three}) imply $h_0=h_2=h_3=0$. Therefore, we obtain $Sym(f) = [0,0,0,0]$, and thus $Sym(f)$ is degenerate. 

Next, assume that $f\in Sig^{(2)}$ with binding coefficients  $(\alpha,\beta)$, which satisfy two equations, (3') $\alpha(a-z)+\beta b=0$ and (4') $\alpha(b-w)+\beta z=0$. Notice that $\alpha$ and $\beta$ cannot be both zero. For simplicity, write $\delta = \frac{\beta}{\alpha}+\frac{b}{a}$. Henceforth, we consider two separate cases. 

\s
[Case: $\alpha\beta=0$] First, assume that $\alpha=0$ and $\beta\neq0$.  {}From (3')--(4'), it follows that 
$f$ should have the form $[a,0,0,d]$.  By a direct calculation of Eq.(\ref{h-123-zero})--(\ref{h-123-three}), we obtain  $Sym(f)=[a^3,0,0,d^3]$. 
Next, assume that $\alpha\neq 0\wedge \beta=0$. Since $f$ must have the form $[a,b,a,b]$ by (3')--(4'), Eq.(\ref{h-123-zero})--(\ref{h-123-three}) imply that $Sym(f) = [A,B,A,B]$, where $A=2a(a^2+3b^2)$ and $B=2b(3a^2+b^2)$. 
In both cases, we conclude that $Sym(f)\in Sig^{(2)}$. 

\s
[Case: $\alpha \beta\neq0\wedge \delta^2=-1$] Since $f$ is symmetric, we can assume that $f=[a,b,z,w]$. Since $\alpha\beta\neq0$, the determinant $\det\tinymatrices{a-z}{b}{b-w}{z}$ equals zero; thus, (5') $z(a-z)=b(b-w)$ follows. 
Now, we set $\gamma =\frac{z}{b}$. It is not difficult to show that (3') implies $\gamma = \frac{\beta}{\alpha}+\frac{b}{a}$, which clearly equals $\delta$.  Now, using (5'), we instantly obtain $z=\delta b$ and $w=-\delta a$. In short, $f=[a,b,\delta b,-\delta a]$ holds. A vigorous calculation of Eq.(\ref{h-123-zero})--(\ref{h-123-three}) shows the following: $h_0=a^3+3ab^2+2\delta b^3$, $h_1=-b(b-\delta a)^2$, $h_2=\delta b(b-\delta a)^2$, and $h_3=\delta(a^3+3ab^2+2\delta b^3)$. Therefore, we conclude that $Sym(f) = [h_0,h_1,\delta' h_1,-\delta' h_0]$, where $\delta'= -\delta$. By its similarity to $f$,  $Sym(f)$ belongs to $Sig^{(2)}$. 
\end{proof}

Concerning the aforementioned signature sets $Sig^{(1)}$ and $Sig^{(2)}$, we define a unique signature set, called $SIG$. To describe this set, we  introduce a new notation $f_{\sigma}$ as follows. Given any ternary signature $f$ and any permutation $\sigma\in S_3$, the notation $f_{\sigma}$ expresses the 
signature $g$ defined by $g(x_1,x_2,x_3) = f(x_{\sigma(1)},x_{\sigma(2)},x_{\sigma(3)})$ for any values $x_1,x_2,x_3\in\{0,1\}$. The $SIG$ is then defined as 
\[
SIG =\{f\mid \forall \sigma\in S_3 [  Sym(f_{\sigma})\not\in\DG \implies Sym(f_{\sigma})\in Sig^{(1)}\cup Sig^{(2)} ] \}.
\]
Our first theorem, Theorem \ref{outside-SIG}, gives a complete classification of the approximation complexity of degree-$2$ $\sharpcspstar$s 
when their signatures fall into {\em outside} of $SIG$.    

\begin{theorem}\label{outside-SIG}
For any ternary signature $f$, if $f\not\in SIG$, then  $\#\mathrm{SAT}_{\complex}\APreduces \sharpcspstar_{2}(f)$. 
\end{theorem}

Since the proof of Theorem \ref{outside-SIG} requires a new notion of T$_{2}$-constructibility, it is postponed until Section \ref{sec:SAT-reduction}.
The theorem makes it sufficient to concentrate only on signatures 
residing within $SIG$. 
To analyze those signatures, we roughly partition $SIG$ into three parts. 
Firstly, we let $SIG_0$ denote the set of all ternary signatures $f$ for which $Sym(f_{\sigma})$ is always 
degenerate for every permutation $\sigma\in S_3$. By Lemma \ref{Sig1-Sym(f)} follows the inclusion $Sig^{(1)}\subseteq SIG_0$. 
Secondly, for each index $i\in\{1,2\}$, let 
$SIG_{i}$ denote the set of all ternary signatures $f$ such that, 
for a certain permutation $\sigma\in S_3$, both 
$Sym(f_{\sigma})\in Sig^{(i)}$ and $Sym(f_{\sigma})\not\in \DG$ hold.  It is obvious that 
$SIG \subseteq SIG_{0}\cup SIG_{1} \cup SIG_{2}$.
Therefore, if we successfully 
classify all degree-$2$ $\sharpcspstar$s whose signatures belong to each of $SIG_{i}$'s,  then we immediately obtain the desired 
complete classification of all degree-$2$ $\sharpcspstar$s. 
Since a whole analysis of $SIG$ seems quite lengthy, this paper 
is focused only on the signature set $SIG_1$, which can be rewritten as  
\[
SIG_{1} = \{ f\mid \exists \sigma\in S_3 \exists a,b\in \complex \;\text{s.t.}\; Sym(f_{\sigma}) = [a,b,-a,-b] \;\&\; a^2+b^2\neq0 \}, 
\]
where the condition $a^2+b^2\neq0$ indicates that $Sym(f_{\sigma})$ is non-degenerate because $rank\tinymatrixthree{a}{b}{-a}{b}{-a}{-b} =rank\tinymatrices{a}{b}{b}{-a} =2$. In what follows, we will describe   
a dichotomy theorem for the associated degree-$2$ $\sharpcspstar$s. 
For ease of notational complication in later sections, we introduce the following useful terminology: a ternary signature $f$ is said to be {\em $SIG_1$-legal} if $Sym(f)$ has the from $[a,b,-a,-b]$ for certain numbers $a,b$ satisfying $a^2+b^2\neq0$. Using this terminology, it follows that $f$ is in $SIG_1$ iff  $f_{\sigma}$ is in $SIG_1$-legal for a certain $\sigma\in S_3$.  


The second theorem---Theorem \ref{main-theorem}---deals with all 
signatures residing within $SIG_1$. 
To state the theorem, however, we need to introduce 
another signature set $\DUP$. 
For our purpose, we begin with a quick explanation of 
the following abbreviation.  For any two ternary signatures $f_0,f_1$, 
the notation $(f_0,f_1)$ expresses the signature $f$ defined as follows: $f(0,x_2,x_3) =f_0(x_2,x_3)$ and $f(1,x_2,x_3)= f_1(x_2,x_3)$ for all pairs  $(x_2,x_3)\in\{0,1\}^{2}$. A vector expression of $f$ makes this definition  simpler; when $f_0=(a,b,c,d)$ and $f_1=(x,y,z,w)$, we obtain $(f_0,f_1)=(a,b,c,d,x,y,z,w)$. 
At last, the basic signature set $\DUP$ is defined as the set of all ternary signatures $f$ such that, after appropriate permutations $\sigma$ of variables, $f_{\sigma}$ becomes of the form $u(x_{\sigma(1)})\cdot (f_0,f_0)$, where $u\in\UU$, and $f_0$ is a certain binary signature. 
We note that $SIG_1\cap\DUP$ is not empty; for instance, 
the signature $f=(1,0,-1,0,i,-2,-i,2)$ is not symmetric but it belongs to both $\DUP$ and $SIG_1$, because $f_{\sigma} = [1,-i](x_1)\cdot (1,0,i,-1,1,0,i,-1)$ and $Sym(f_{\sigma}) = 7\cdot[1,-1,-1,1]$ for $\sigma=(x_2x_1x_3)$, where $i=\sqrt{-1}$. 
Two examples of important signatures in $\DUP$ include: $f=(0,0,0,0,x,y,z,w)$ and $f=(a,b,c,d,0,0,0,0)$. 

Finally, the second classification theorem is stated as follows. 

\begin{theorem}\label{main-theorem}
Let $f$ be any ternary signature in $SIG_1$. If 
$f$ is in $\DUP$, then $\sharpcspstar_{2}(f)$ is in $\fp_{\complex}$. Otherwise, $\#\mathrm{SAT}_{\complex}$ is AP-reducible to $\sharpcspstar_{2}(f)$. 
\end{theorem}

Theorem \ref{main-theorem} follows from three key propositions,  Propositions \ref{SAT-ONE3-EQ2}--\ref{h0+h2=h1+h3=0}, which  
will be explained in Section \ref{sec:T-const-tech}, and the proof of Theorem \ref{main-theorem} will be presented in Section \ref{sec:key-props}. 

\section{T$_{2}$-Constructibility Technique}\label{sec:T-const-tech}

To prove our main theorems stated in Section \ref{sec:main-theorem},  we intend to employ two new technical tools. In this section, 
we will introduce the first technical tool, called {\em T$_{2}$-constructibility}.  
Applying this technical tool to degree-$2$ $\sharpcspstar$s with a help of three supplemental propositions, Propositions \ref{SAT-ONE3-EQ2}--\ref{h0+h2=h1+h3=0}, we will be able to give the proof of the main theorems.

\subsection{T$_{2}$-Constructibility}\label{sec:T-const}

When we wish to calculate approximate solutions of degree-$2$ $\sharpcspstar$s, in place of the exact solutions, 
standard tools like ``polynomial interpolation'' are no longer applicable.  A useful tool in  determining the approximation complexity of unbounded-degree $\sharpcspstar$'s used in \cite{Yam10a} is the notion of {\em T-constructibility}.  
Because degree-$2$ $\sharpcspstar$s are quite different from 
unbounded-degree $\sharpcspstar$s, 
its appropriate modification is needed to meet our requirement. 

\sloppy To pursue notational succinctness, we use the following notations. 
For any index $i\in[k]$ and any bit $c\in\{0,1\}$, 
the notation $f^{x_i=c}$ denotes the function $g$ satisfying that $g(x_1,\ldots,x_{i-1},x_{i+1},\ldots,x_k) = f(x_1,\ldots,x_{i-1},c,x_{i+1},\ldots,x_k)$.  Similarly, 
let $f^{x_i=*}$ express the function $g$ defined as  $g(x_1,\ldots,x_{i-1},x_{i+1},\ldots,x_k) = \sum_{x_i\in\{0,1\}}f(x_1,\ldots,x_{i-1},x_i,x_{i+1},\ldots,x_k)$.
When two indices $i,j\in[k]$ satisfy $i<j$, we write $f^{x_i=x_j=*}$ 
for the function $g$ defined as  $g(x_1,\ldots,x_{i-1},x_{i+1},\ldots, x_{j-1},x_{j+1},\ldots,x_k) = \sum_{x_i\in\{0,1\}} f(x_1,\ldots,x_{i-1},x_i,x_{i+1},\ldots, x_{j-1},x_{i},x_{j+1},\ldots,x_k)$, where the second $x_i$ appears at the $j$th position. 
Moreover, let $(g_1\cdot g_2)(x_1,\ldots,x_k,y_1,\ldots,y_{k'}) = g_1(x_1,\ldots,x_k) g_2(y_1,\ldots,y_{k'})$ whenever $g_1$ and $g_2$ take ``disjoint'' sets of variables $\{x_1,\ldots,x_k\}$ and $\{y_1,\ldots,y_{k'}\}$, respectively. 
In a similar way, $\lambda\cdot g$ is defined as $(\lambda \cdot g)(x_1,\ldots,x_k) = \lambda\cdot g(x_1,\ldots,x_k)$. 

We say that a signature $f$ of arity $k$ is {\em T$_{2}$-constructible} (or {\em T$_{2}$-constructed}) from a set $\GG$ of signatures if $f$ can be obtained, initially from signatures in $\GG$, by recursively applying a finite number (possibly zero) of operations described below.
\begin{enumerate}
\item {\sc Permutation:} for two indices $i,j\in[k]$ with $i<j$, 
by exchanging two columns $x_i$ and $x_j$, we transform $g$ into $g'$ that  is defined by $g'(x_1,\ldots,x_i,\ldots,x_j,\ldots,x_k) = g(x_1,\ldots,x_j,\ldots,x_i,\ldots,x_k)$.
\vs{-2}
\item {\sc Pinning:} for an index $i\in[k]$ and a bit $c\in\{0,1\}$, 
we build $g^{x_i=c}$ from $g$.
\vs{-2}
\item {\sc Projection:} for an index $i\in[k]$,  
we build $g^{x_i=*}$ from $g$.
\vs{-2}
\item {\sc Linked Projection:} for two indices $i,j\in[k]$ with $i<j$, we build $g^{x_i = x_j =*}$ from $g$.
\vs{-2}
\item {\sc Expansion:}  for an index $i\in[k]$, we introduce a new ``free'' variable, say, $y$ and  transform $g$ into $g'$, which is defined by $g'(x_1,\ldots,x_i,y,x_{i+1},\ldots,x_k) =  g(x_1,\ldots,x_{i},x_{i+1},\ldots,x_k)$.  
\vs{-2}
\item {\sc Exclusive Multiplication:}  from two signatures $g_1$ of arity $k$ and  $g_2$ of arity $k'$, if $g_1$ and $g_2$ take disjoint variable sets, then we build $g_1\cdot g_2$ from $\{g_1,g_2\}$.  
\vs{-2}
\item {\sc Normalization:}  for a constant $\lambda\in\complex-\{0\}$, 
we build $\lambda\cdot g$ from $g$. 
\end{enumerate}
Main features of T$_{2}$-constructibility are two special operations:  linked projection and exclusive multiplication. These operations 
reflect the structure of a signature grid, and therefore they are quite different from their associated operations used for the T-constructibility. 
When $f$ is T$_{2}$-constructible from $\GG$, we use the notation  $f\leq_{con}^{*}\GG$; in particular, when $\GG=\{g\}$, 
we simply write $f\leq_{con}^{*}g$ 
instead of $f\leq_{con}^{*}\{g\}$. 

The most useful claim at this moment is  the T$_{2}$-constructibility of $Sym(f)$ from $f$, and we state this claim as a lemma for later referencing. 

\begin{lemma}\label{Sym(f)-reduce-f}
For any ternary signature $f$, it holds that $Sym(f)\leq_{con}^{*}f$. 
\end{lemma}

\begin{proof}
To T$_{2}$-construct $Sym(f)$ from $f$, we first generate a product of 
$f(x_1,x_2,z_2)$, $f(y_1,y_2,x'_2)$, and $f(z_1,z'_2,y'_2)$ using  {\sc Exclusive Multiplication} with all distinct variables. We then apply {\sc Linked Projection} by identifying $x'_2,y'_2,z'_2$ with $x_2,y_2,z_2$, respectively. 
\end{proof}

The following lemma bridges between the T$_{2}$-constructibility and the AP-reducibility. 

\begin{lemma}\label{T-const-reduction}
Let $f$ be any signature and let $\FF,\GG$ be any two signature sets. 
If $f\leq_{con}^{*}\GG$, then $\sharpcspstar_{2}(f,\FF)\APreduces \sharpcspstar_{2}(\GG,\FF)$.
\end{lemma}

\begin{proof}
Our proof is similar in nature to the T-constructibility proof of \cite[Lemma 5.2]{Yam10a}. All operations except for {\sc Expansion}, {\sc Linked Projection}, and {\sc Exclusive Multiplication} can be handled in such a way similar to 
the case of the T-constructibility. 
Therefore, in what follows, we will show the lemma for those three exceptional operations. Now, let $\FF$ denote any signature set and let $\Omega=(G,\FF',\pi)$ express any signature grid given as input instance to $\sharpcspstar_{2}(f,\FF)$. 

\s
{\sc [Expansion]}\hs{1} 
For simplicity, let $f(y,x_1,\ldots,x_k) = g(x_1,\ldots,x_k)$, where $y$ is a new free variable. Let us consider a subgraph $G'$ of $G$ such that it consists of node $v$ labeled $f$ and node $w$ adjacent to $v$ by an edge labeled $y$. Now, we want to define a new subgraph $\tilde{G}'$ to replace $G'$.  
First, we remove the edge $y$ so that we split $G'$ into two disconnected subgraphs. Second, we replace the node $v$ by a new node $v'$ whose label is $g$. Third, we insert a new node $u$ with label $[1,1]$ between the two nodes $v'$ and $w$ by two new edges. Let $\Omega'$ be obtained from $\Omega$ by applying this modification to all nodes with the label $f$. It thus holds that $\holant_{\Omega} = \holant_{\Omega'}$.  This leads to $\sharpcspstar_{2}(f,\FF)\APreduces \sharpcspstar_{2}(g,\FF)$. 

\s
{\sc [Linked Projection]}\hs{1}
Let $f= g^{x_i=x_j=*}$. To improve readability, we assume that $i=1$ and $j=2$; that is, $f(x_3,\ldots,x_k) = \sum_{x_1\in\{0,1\}} g(x_1,x_2,x_3,\ldots,x_k)$. We are focused on node $v$ labeled $f$ in $G$. Let us consider a subgraph $G'$ consisting of this node $v$ and all the other nodes adjacent to $v$. We replace $G'$ by another graph $\tilde{G}'$  that is defined as follows. First, we replace the label $f$ of the node $v$ with $g$. Second, we add a new edge $(v,v)$. Now, define $\Omega'$ as the signature grid obtained by replacing $G'$ with $\tilde{G}'$. It is not difficult to show that $\holant_{\Omega} = \holant_{\Omega'}$. Therefore, if we recursively replace all nodes labeled $f$, we finally obtain an AP-reduction:  $\sharpcspstar_{2}(f,\FF)\APreduces \sharpcspstar_{2}(g,\FF)$. 

\s
{\sc [Exclusive Multiplication]}\hs{1} 
For two disjoint sets of variables $\{x_1,x_2,\ldots,x_k\}$ and $\{y_1,\ldots,y_{k'}\}$, we  assume that $g_1$ and $g_2$ take variable series  $(x_1,\ldots,x_k)$ and $(y_1,\ldots,y_{k'})$, respectively, and let $f = g_1\cdot g_2$. Now, we consider a subgraph $G'$ that contains node $v$ labeled $f$ and all the other nodes adjacent to $v$. We wish to define a new subgraph $\tilde{G}'$ as follows. 
First, we split $G'$ into two subgraphs $G'_1$ and $G'_2$, where $G'_1$ (resp., $G'_2$) is obtained from $G'$ by deleting the edges $y_1,\ldots,y_{k'}$ (resp., $x_1,\ldots,x_k$) as well as all nodes, except  for $v$, attached to those edges.  In the subgraph $G'_1$ (resp., $G'_2$), we replace the node $v$ by a new node $v'_1$ (resp. $v'_2$) with the label $g_1$ (resp., $g_2$). After eliminating all nodes with the label $f$ in this way, we finally obtain  from $\Omega$  a signature grid, say, $\Omega'$. The equation  $\holant_{\Omega}=\holant_{\Omega'}$ easily follows, and we then obtain 
$\sharpcspstar_{2}(f,\FF)\APreduces \sharpcspstar_{2}(g_1,g_2,\FF)$. 
\end{proof}

By a direct application of Lemma \ref{T-const-reduction} with Lemma \ref{Sym(f)-reduce-f} to $Sym(f)$, it immediately follows that   $\sharpcspstar_{2}(Sym(f),\FF)\APreduces \sharpcspstar_{2}(f,\FF)$ for any signature set $\FF$. This simple fact is actually a key to our main theorems, which will be proven in the subsequent subsections.  

\subsection{\#SAT$_{\complex}$-Hardness under AP-Reducibility}\label{sec:SAT-reduction}

When dealing with all complex numbers, Turing reducibility does not always induce AP-reducibility; as a result,  
the computational hardness of a counting problem under Turing reducibility may not immediately result in its computational hardness under AP-reducibility. 
Since there has been little work on the approximation complexity of Holant problems, there is no written proof for the fact that $\#\mathrm{SAT}_{\complex}\APreduces \holantstar(g_1|g_2)$ for every $(g_1|g_2)\in\BB$.  
To use Lemmas \ref{cai-lemmas} in our setting of approximation complexity, we first need to establish this hardness result of $\holantstar(g_1|g_2)$   under AP-reductions. 

\begin{proposition}\label{SAT-ONE3-EQ2}
For every pair $(g_1|g_2)\in\BB$, it holds that $\#\mathrm{SAT}_{\complex}\APreduces \holantstar(g_1|g_2)$. 
\end{proposition}

\begin{proof}
First, we show that $\#\mathrm{SAT}_{\complex}\APreduces \holantstar(EQ_3|OR)$. 
Now, let us recall a few known results from \cite{Yam10a,Yam10b}.  
It is known that $\#\mathrm{SAT}_{\complex} \APreduces \sharpcspstar(OR)$ \cite{Yam10a} and that  $\sharpcspstar(OR)  \APequiv \sharpcspstar_{3}(OR) \APequiv \holantstar(EQ_3|OR)$ \cite{Yam10b}. 
Combining these results, we conclude that $\#\mathrm{SAT}_{\complex} \APreduces \holantstar(EQ_3|OR)$. 

Next, we show that $\#\mathrm{SAT}_{\complex} \APreduces \holantstar(ONE_3|EQ_2)$. Let $f=Sym(ONE_3)$ for brevity.  Our proof is made up of five steps. Recall that all signatures in this paper are represented as row vectors. 

(1) By a simple calculation, we obtain $f=[4,2,1,1]$.  Since $f\leq_{con}^{*}ONE_3$, By Lemma \ref{T-const-reduction} implies that $\holantstar(f|EQ_2)\APreduces \holantstar(ONE_3|EQ_2)$. 

(2) Let $M=\tinymatrices{a}{b}{c}{d}$, where $a,b,c,d\in\complex$ are defined later. We consider a holographic transformation from $\holantstar(f|EQ_2)$ to $\holant(EQ_3|g)$ for a certain binary signature $g$. To make this transformation possible, $M$ needs to satisfy that $f = EQ_3 M^{\otimes 3}$ and $g^{T} = M^{\otimes 2} EQ_2^{T}$. With this $M$, Lemma \ref{holographic-transform} establishes the AP-equivalence:  $\holantstar(f|EQ_2)\APequiv \holantstar(EQ_3|g)$. 
Note that $EQ_3 M^{\otimes 3} = [a^3+c^3,a^2b+c^2d,ab^2+cd^2,b^3+d^3]$. Since $f=[4,2,1,1]$, we obtain $a^3+c^3=4$, $a^2b+c^2d=2$, $ab^2+cd^2=1$, and $b^3+d^3=1$. Here, we consider the case of $a=2b$. Since  $a^3+c^3=4$,  we obtain $a^3+c^3=2(a^2b+c^2d)$, which implies $c^2(c-2d)=0$. Now, 
we claim that $c=0$. Assuming otherwise, we obtain $c=2d$, which yields $a^3+c^3=8(b^3+d^3)=4$. Thus,  $b^3+d^3\neq 1$ follows; this is a contradiction. Hence, it must hold that $c=0$. With this $c$,  $a^3+c^3=4$ implies $b^3=1/2$, and $b^3+d^3=1$ also implies $d^3=1/2$. Overall, it suffices to we define  $M$ as $\gamma\tinymatrices{2}{1}{0}{1}$, where   $\gamma =(1/2)^{1/3}$. 

(3) Since $g^{T} = M^{\otimes 2} EQ_2^{T}$, $g$ equals $\gamma^2 \cdot (5,1,1,1)$. As discussed in Section \ref{sec:randomized-scheme}, it holds that $\holantstar(EQ_3|g)\APequiv \sharpcspstar(g)$; thus, we obtain  $\sharpcspstar(g)\APreduces \holantstar(ONE_3|EQ_2)$.  

(4) We want to show that $\sharpcspstar(OR)\APreduces \sharpcspstar(g)$. 
In this step, we use the notion of T-constructibility \cite{Yam10a}. Let $g'=[5,1,1]$ so that $g'\leq_{con}^{*}g$. Now, define $h(x,y) = -(1/4)\sum_{z\in\{0,1\}} g'(x,z)g'(z,y)u(z)$, where $u=[1,-25]$. It is not difficult to show that $h=[0,5,6]$. Since $h$ is T-constructible from 
$\{g',u\}$, by applying a result of \cite[Lemma 5.2]{Yam10a}, we obtain $\sharpcspstar(h)\APreduces \sharpcspstar(g')$. It is also shown in \cite[Lemma 6.4]{Yam10a} that $\sharpcspstar(OR)\APreduces \sharpcspstar([0,u,v])$ for any constants $u,v\in\complex-\{0\}$. Hence, we conclude that $\sharpcspstar(OR)\APreduces \sharpcspstar(h)$.   

(5) Since $\#\mathrm{SAT}_{\complex}\APreduces \sharpcspstar(OR)$, we finally establish the desired AP-reduction: $\#\mathrm{SAT}_{\complex}\APreduces \holantstar(ONE_3|EQ_2)$.  
\end{proof}

We are now ready to prove the first main theorem, Theorem  \ref{outside-SIG}. 
Proposition \ref{SAT-ONE3-EQ2} greatly simplify the proof of the theorem.  

\begin{proofof}{Theorem \ref{outside-SIG}}
Let $f$ be any ternary signature not in $SIG$; namely, there exists a permutation $\sigma\in S_3$ for which  $Sym(f_{\sigma})\not\in Sig^{(1)}\cup Sig^{(2)}$ and $Sym(f_{\sigma})\not\in \DG$. With the help of Proposition \ref{SAT-ONE3-EQ2}, Lemma \ref{cai-lemmas}(\ref{cai-third-case}) leads to the conclusion that 
$\#\mathrm{SAT}_{\complex}\APreduces \holantstar(EQ_2|Sym(f_{\sigma}))$. 
By Lemma \ref{holant-reduction}(2), it follows that $\holantstar(EQ_2|Sym(f_{\sigma})) \APreduces \sharpcspstar_{2}(Sym(f_{\sigma}))$. 
Since $Sym(f_{\sigma})\leq_{con}^{*} f_{\sigma}$ by Lemma \ref{Sym(f)-reduce-f},  Lemma \ref{T-const-reduction} implies that  $\sharpcspstar_{2}(Sym(f_{\sigma}))\APreduces 
\sharpcspstar_{2}(f_{\sigma})$. 
Finally, because $\sharpcspstar_{2}(f_{\sigma})$ and 
$\sharpcspstar_{2}(f)$ are AP-equivalent to each other, we immediately obtain $\#\mathrm{SAT}_{\complex}\APreduces \sharpcspstar_{2}(f)$, as required. 
\end{proofof}

\subsection{Two Key Propositions}\label{sec:key-props} 

The proof of Theorem \ref{main-theorem} is composed of three propositions.  The first proposition---Proposition \ref{SAT-ONE3-EQ2}---has already proven in Section \ref{sec:SAT-reduction}. The second proposition below concerns the computability result of degree-$2$ $\sharpcspstar$s whose signatures are all drawn from $\DUP$. For completeness, we include the proof of this proposition. 

\begin{proposition}\label{computability-result}
For any subset $\FF\subseteq \DUP$, it holds that $\sharpcspstar_{2}(\FF)$ is in $\fp_{\complex}$. 
\end{proposition}

\begin{proof}
Let $\FF\subseteq \DUP$. 
We demonstrate how to solve the counting problem $\sharpcspstar_{2}(\FF)$ in polynomial time. 
Let $\Omega=(G,\FF',\pi)$ be any input signature grid to $\sharpcspstar_{2}(\FF)$. 
Our proof proceeds by induction on the number of degree-$3$ nodes in $G$. 
We recursively ``break down'' ternary signatures into binary ones. 
Let us consider the base case:  all nodes are of degree $1$.  
We conveniently express a binary signature $f=(a,b,c,d)$ as $\tinymatrices{a}{b}{c}{d}$. 

[Case 1] Consider the case where all nodes are of degree $1$; thus, $G$ consists of disconnected subgraphs, each of which is composed of two degree-$1$ nodes connected by one edge. For each $G'$ of such subgraphs, let $\Omega'$ denote its associated signature grid. If $G'$ contains two nodes labeled $f=(a,b)$ and $g=(x,y)$, then the value $\holant_{\Omega'}$ equals $(a\;\;b)\tinycomb{x}{y}$. The whole $\holant_{\Omega}$ then 
is calculated as the product of $\holant_{\Omega'}$ over all possible $\Omega'$'s. 
The computation time of $\holant_{\Omega}$ is obviously proportional to the number of $\Omega'$'s.  

[Case 2] Assume that all nodes are of degrees at most $2$. In a recursive way, we wish to replace nodes of degree $2$ by nodes of degree $1$. In the end, all remaining nodes become degree $1$. This recursive process halts after steps less than or equal to the number of nodes in $G$. 
Now, we choose a node $f_1$ of degree $2$ and assume that node $f_1$ has two edges  $e_1=(f_1,f_2)$ and $e_2=(f_1,f_3)$, where $f_2$ and $f_3$ are nodes of degrees at most $2$. Let $f_1=(a,b,c,d)$. By permuting $e_1$ and $e_2$, without loss of generality, we may assume that an instance to $f_1$ has the form $(e_2,e_1)$. 
Consider a subgraph $G'$ consisting of the nodes $f_1$ and $f_2$ and the  
edge $e_1$.   

(1) Assume that the node $f_2$ has degree $1$ and let $f_2=(x,y)$. We introduce a new signature $f'=\tinymatrices{a}{b}{c}{d}(x\;\;y)$ over the  variable $e_2$. Finally, we replace $G'$ by a node with label $f'$. Let $\Omega'$ be the signature grid obtained from this replacement. It is not difficult to show that $\holant_{\Omega}=\holant_{\Omega'}$. 

(2) Next, we assume that the node $f_2$ is of degree $2$ and assume that   $f_2=(a,y,z,w)$ takes a variable series $(e_1,e_3)$, where $e_3$ is another edge.  A new signature $f'$ is defined as $\tinymatrices{a}{b}{c}{d} \tinymatrices{x}{y}{z}{w}$. We then replace $G'$ by a node labeled $f'$. This replacement does not change the 
value $\holant_{\Omega}$.    

[Case 3] 
We assume that certain nodes still have degree $3$. We recursively replace each node of degree $3$ by two nodes of degree $2$ and of degree $1$. 
First, choose a node $f_1$ of degree $3$ and assume that $f_1$ has edges $e_1=(f_1,f_2)$, $e_2=(f_1,f_3)$, and $e_3=(f_1,f_4)$. Since $f_1\in \DUP$, $f_1$ has the form $u(x_1)\cdot (f_0,f_0)$, where $f_0$ is of arity $2$. Next, we consider a subgraph $G'$ made up of four nodes labeled  $f_1,f_2,f_3,f_4$ and four edges $e_1,e_2,e_3,e_4$. We then delete the edge $e_1$ from $G'$ and split $G'$ into two disconnected subgraphs, say, $G_1$ and $G_2$. Assume that $G_1$ consists of the node $f_2$ and $G_2$ consists of three nodes $f_1,f_3,f_4$. For $G_1$, we prepare a new node labeled $u$ and attach it to node $f_2$ by a new edge $e'_1$. For $G_2$, we replace the node $f_1$ by the node $f_0$. Let $\Omega'$ be the signature grid obtained by this modification. It is not difficult to show that $\holant_{\Omega} = \holant_{\Omega'}$. 
\end{proof}

Finally, we state the third proposition, which gives a crucial property of signatures in $SIG_1$.   

\begin{proposition}\label{h0+h2=h1+h3=0}
Let $f$ be an arbitrary signature in $SIG_{1}$. 
If $f$ is not in $\DUP$, then 
there exists a non-degenerate symmetric signature $g=[g_0,g_1,g_2]$ such that $g\leq_{con}^{*}\GG\cup\{f\}$, where $\GG$ is a finite subset of $\UU$, and $(g_0\neq g_2\vee g_1\neq 0)\wedge g_0+g_2\neq 0$. 
\end{proposition}

With a use of Propositions \ref{SAT-ONE3-EQ2}--\ref{h0+h2=h1+h3=0}, 
 Theorem \ref{main-theorem} can be succinctly proven below. 

\begin{proofof}{Theorem \ref{main-theorem}}
Let $f$ be any ternary signature in $SIG_1$. If $f$ is in $\DUP$, then 
Proposition \ref{computability-result} imposes  $\sharpcspstar_{2}(f)$ to be inside $\fp_{\complex}$. Next, we assume that $f\not\in \DUP$. 
By Proposition \ref{h0+h2=h1+h3=0}, there exists a non-degenerate symmetric  binary signature $g$ such that $g$ is either not of the form $[a,b,-a]$ or not of the form $[a,0,a]$ for any numbers $a,b\in\complex$. 
This $g$ is obviously T$_{2}$-constructed from $\GG\cup\{f\}$, where $\GG$ is  a finite subset of $\UU$.  Hence, it follows by Lemma \ref{T-const-reduction} that  $\sharpcspstar_{2}(f,g)\APreduces \sharpcspstar_{2}(f)$. 
Moreover,  Lemma \ref{cai-lemmas}(\ref{cai-first-case}) ensures the existence of a pair $(g_1|g_2)\in\BB$ satisfying that  $\holantstar(g_1|g_2)\APreduces \holantstar(EQ_2|f,g)$. 
Proposition \ref{SAT-ONE3-EQ2} shows that   $\#\mathrm{SAT}_{\complex}^{*}\APreduces \holant(g_1|g_2)$. By Lemma \ref{holant-reduction}(2),  $\holantstar(EQ_2|f,g)\APreduces \sharpcspstar_{2}(f,g)$ also holds. 
Combining those AP-reductions, 
we conclude that $\#\mathrm{SAT}_{\complex}^{*}\APreduces  \sharpcspstar_{2}(f)$, as requested. 
\end{proofof}

Now, the remaining task is to prove Proposition \ref{h0+h2=h1+h3=0} and the rest of this paper is devoted to giving its proof. For our purpose, we will need another new idea, called {\em parametrized symmetrization}.

\section{Parametrized Symmetrization Technique}\label{sec:rparametric-sym}

We have shown in Section \ref{sec:asymmetric-main} how to transform arbitrary ternary signatures into symmetric ternary signatures. To prove Proposition \ref{h0+h2=h1+h3=0}, 
we also need to produce symmetric ``binary'' signatures from arbitrary 
``ternary'' signatures so that we can make use of 
Lemma \ref{cai-lemmas}(\ref{cai-first-case}).  
Here, we will introduce the second scheme of symmetrization, which 
is quite different from the first scheme given in Section \ref{sec:asymmetric-main}; 
in fact, this new scheme is  ``parametrized.'' In other words, it is not a fixed symmetrized signature as in Eq.(\ref{Sym(f)-def}); instead, it consists of an ``infinite series'' of symmetrized signatures. In this section, we assume that our target ternary signature $f$ has the form $(a,b,c,d,x,y,z,w)$. 
Later in Section \ref{sec:outline-proof}, we will give the proof of Proposition \ref{h0+h2=h1+h3=0}. 

\subsection{Parametrized Symmetrization Scheme}\label{sec:first-symmetrization}

A {\em parametrized symmetrization scheme} produces a set of degree-$2$ polynomials. This scheme is simple and easy to apply in the proof of Proposition \ref{h0+h2=h1+h3=0}. We first fix an arbitrary unary signature $u$ and we introduce 
$SymL(f)$ as a new signature defined as 
\[
SymL(f)(x_2,y_2) = \sum_{x_1,x_3,y_1\in\{0,1\}} f(x_1,x_2,x_3) f(y_1,y_2,x_3) u(x_1) u(y_1).
\]
It is important to note that $SymL(f)\leq_{con}^{*} \{f,u\}$. 
A simple calculation shows that, in particular, when $u=[0,1]$, 
$SymL(f)$ equals $[x^2+y^2,xz+yw,z^2+w^2]$.
In contrast, when $u=[1,\varepsilon]$ for a complex value $\varepsilon$,  $SymL(f)=[g_{0},g_{1},g_{2}]$ satisfies:   
\begin{enumerate}
\item $g_0 = \varepsilon^2(x^2+y^2) + 2\varepsilon(ax+by) + a^2+b^2$,
\vs{-2}
\item $g_1 = \varepsilon^2(xz+yw) + \varepsilon(az+bw+cx+dy) + ac+bd$, and
\vs{-2}
\item $g_2 = \varepsilon^2(z^2+w^2) + 2\varepsilon(cz+dw) + c^2+d^2$.
\end{enumerate}

In the rest of this paper, we fix $u=[1,\varepsilon]$. To emphasize the parameter $\varepsilon$ inside $u$, we also write $SymL(f)_{\varepsilon}$ and $[g_{0,\varepsilon},g_{1,\varepsilon},g_{2,\varepsilon}]$. 
One of the most important and useful properties is the non-degeneracy of $SymL(f_{\sigma})_{\varepsilon}$. Here, we prove that, when $f$ does not belong to $\DUP$, $SymL(f)$ cannot be a degenerate signature.

\begin{proposition}\label{SymL-non-degenerate}
Let $f$ be any ternary signature. 
If $f\not\in\DUP$, then $SymL(f_{\sigma})_{\varepsilon}$ is non-degenerate for any permutation $\sigma\in S_3$ 
and for all but finitely many numbers $\varepsilon\in\complex$. 
\end{proposition}

Since the proof of this proposition demands fundamental properties of $SymL(f)$ that are listed in Section \ref{sec:SymL(f)},  
we postpone the proof until Section \ref{sec:proof-SymL}.

\subsection{Proof of Proposition \ref{h0+h2=h1+h3=0}}\label{sec:outline-proof}

In Sections \ref{sec:asymmetric-main} and \ref{sec:first-symmetrization}, we have introduced two schemes of symmetrization. These schemes are powerful enough to prove Proposition \ref{h0+h2=h1+h3=0}, 
which is a basis of the proof of Theorem \ref{main-theorem}.  
Henceforth, we will present the proof of Proposition \ref{h0+h2=h1+h3=0}. Our goal is 
to prove that, for a given ternary signature $f$ in $SIG_1$,  if $f\not\in\DUP$, then $SymL(f_{\sigma})_{\varepsilon}$ becomes the desired $g$ stated in the proposition for certain values  of 
$\sigma$ and $\varepsilon$. 
We proceed our argument by way of contradiction. Let us describe this argument in more details.     

\s

Let $f$ be any ternary signature not in $\DUP$. Without loss of generality, we fix a permutation $(x_1x_2x_3)$ and assume that $Sym(f)$ is non-degenerate and is $SIG_1$-legal. 
For any given permutation $\sigma\in S_3$, 
we write $SymL(f_{\sigma})_{\varepsilon} = [g_{0,\varepsilon}^{\sigma},g_{1,\varepsilon}^{\sigma},g_{2,\varepsilon}^{\sigma}]$, as done in Section \ref{sec:first-symmetrization}.   
Hereafter, we want to prove that there exists a permutation $\sigma\in S_3$ 
such that both $g_{0,\varepsilon}^{\sigma}+ g_{2,\varepsilon}^{\sigma}\neq0$ and $g_{0,\varepsilon}^{\sigma}\neq g_{2,\varepsilon}^{\sigma}\vee g_{1,\varepsilon}^{\sigma}\neq0$ hold for all but finitely many values  $\varepsilon\in\complex$. Now, assume otherwise; that is, 
\begin{itemize}
\item[(*)] for every permutation $\sigma$ and for all but finitely many values of $\varepsilon$, either (i) $g_{0,\varepsilon}^{\sigma} +g_{2,\varepsilon}^{\sigma}=0$ or (ii) $g_{0,\varepsilon}^{\sigma} =g_{2,\varepsilon}^{\sigma}\wedge g_{1,\varepsilon}^{\sigma}=0$ holds. 
\end{itemize}
We first note that the above two conditions (i) and (ii) do not hold simultaneously. 
To see this, assume that the two conditions hold together; thus,   
$g_{0,\varepsilon}^{\sigma} =g_{1,\varepsilon}^{\sigma} =g_{2,\varepsilon}^{\sigma} =0$ follows. In short, it holds that  $SymL(f_{\sigma})_{\varepsilon}=[0,0,0]$. This clearly 
indicates the degeneracy of  $SymL(f_{\sigma})_{\varepsilon}$, 
contradicting Proposition \ref{SymL-non-degenerate}. Therefore, exactly one of the two conditions should hold. This fact will be frequently used in Sections \ref{sec:proof-SymL}--\ref{sec:type-III}. 

Our assumption (*) can be nailed down to the following three cases so that each case can be discussed separately. 
First, let us consider the case where the condition (ii) always holds for every permutation $\sigma$ and for almost all values of $\varepsilon$. 
For each fixed $\sigma\in S_3$, since the equations $g_{0,\varepsilon}^{\sigma}= g_{2,\varepsilon}^{\sigma}$ and $g_{1,\varepsilon}^{\sigma} =0$ can be viewed as a set of polynomial equations in $\varepsilon$ of degrees at most two, the condition (ii) fails for at most two values of $\varepsilon$. 
Since $f$ is $SIG_1$-legal, this case obviously contradicts the consequence of Proposition \ref{prop-type-I} given below. For readability, we postpone the proof of this proposition until Section \ref{sec:type-I}. 
\begin{proposition}\label{prop-type-I}
Let $f$ be any ternary signature not in $\DUP$. If $f$ is $SIG_1$-legal, 
then there exists a permutation $\sigma$  such that 
either $g_{0,\varepsilon}^{\sigma} \neq  g_{2,\varepsilon}^{\sigma}$ or  $g_{1,\varepsilon}^{\sigma} \neq 0$ holds for at least three  
distinct values of $\varepsilon$.
\end{proposition}

Next, let us consider the case where two distinct permutations $\sigma$ and $\tau$ satisfy the conditions (i) and (ii), respectively, for almost all values of $\varepsilon$.  
As the following proposition indicates, Statement (*) forces this case to fail.  The proposition will be proven in Section \ref{sec:type-II}. 
\begin{proposition}\label{prop-type-II}
Let $f$ be any ternary signature such that $f$ is $SIG_1$-legal. 
Assume that $f\not\in\DUP$. If Statement (*) holds, then the following property is never satisfied: there are two distinct permutations $\sigma$ and $\tau$ for which  $g_{0,\varepsilon}^{\sigma}=g_{2,\varepsilon}^{\sigma}\wedge g_{1,\varepsilon}^{\sigma}=0$ and $g_{0,\varepsilon}^{\tau}+g_{2,\varepsilon}^{\tau}=0$ for all but finitely many values of $\varepsilon$.   
\end{proposition}

Finally, we consider the remaining situation that the condition (i) holds  for every permutation $\sigma$ and for almost all values of $\varepsilon$.   Proposition \ref{prop-type-III} implies that $f\in\DUP$; however, this contradicts our assumption that $f\not\in\DUP$. In  Section \ref{sec:type-III}, we will give the proof of this proposition. 
\begin{proposition}\label{prop-type-III}
Let $f$ be any ternary signature that is $SIG_1$-legal. 
Assume that, for every permutation $\sigma\in S_3$ and for all but finitely many $\varepsilon$'s , $g_{0,\varepsilon}^{\sigma}+ g_{2,\varepsilon}^{\sigma}=0$ holds. It then holds that $f\in\DUP$.   
\end{proposition}

Since all the above three cases lead to contradictions, we then conclude that 
Statement (*) does not hold. Hence,   
there exist a permutation $\sigma\in S_3$  and a value 
$\varepsilon\in\complex$ for which  $g_{0,\varepsilon}^{\sigma}+ g_{2,\varepsilon}^{\sigma}\neq0$ and $g_{0,\varepsilon}^{\sigma}\neq g_{2,\varepsilon}^{\sigma}\vee g_{1,\varepsilon}^{\sigma}\neq0$. 
Choose such a pair $(\sigma,\varepsilon)$ and  define the desired $g$ (stated in Proposition \ref{h0+h2=h1+h3=0}) to be $SymL(f_{\sigma})_{\varepsilon}$. 
Notice that, since $f\not\in\DUP$, Proposition \ref{SymL-non-degenerate} guarantees the non-degeneracy of $g$.  
Therefore, the proof is now completed.

\section{Fundamental Properties of Symmetrization Schemes}\label{sec:SymL(f)}

To simplify proofs that will be given in Sections \ref{sec:proof-SymL}--\ref{sec:type-III}, we wish to list useful properties, equations, and conditions that fulfill the requirements 
of $Sym(f)$ as well as $SymL(f)$.  Throughout this section, we fix a ternary signature $f=(a,b,c,d,x,y,z,w)$. 

In the subsequent subsections, we will take the following convention. A permutation $\sigma$ in $S_3$ should be formally 
expressed as, \eg  $\sigma= (312)$;  for clarity, we 
slightly abuse this notation and treat it as a permutation over three different variables $x_1,x_2,x_3$. Thus, we write $\sigma=(x_3x_1x_2)$ instead of $\sigma=(312)$ to stress the central roles of those variables.

\subsection{Basic Properties of SymL(f)}

Let us consider the parametrized symmetrization $SymL(f)_{\varepsilon} =[g_{0,\varepsilon}^{\sigma},g_{1,\varepsilon}^{\sigma},g_{2,\varepsilon}^{\sigma}]$ of $f$. We want to present necessary conditions for three different situations in which each of the following holds:  (i) $g_{0,\varepsilon}^{\sigma}+g_{2,\varepsilon}^{\sigma} =0$, (ii) $g_{0,\varepsilon}^{\sigma} =g_{2,\varepsilon}^{\sigma} \wedge g_{1,\varepsilon}^{\sigma}=0$, and (iii)  $g_{0,\varepsilon}^{\sigma}g_{2,\varepsilon}^{\sigma} = (g_{1,\varepsilon}^{\sigma})^2$. The parameter $\varepsilon$ tends to be omitted whenever it is clear from the context.

\subsubsection{Situation 1: $g_0+g_2=0$}\label{sec:SymL-sit-2}
 
Meanwhile, we 
fix $\sigma=(x_1x_2x_3)$ and omit subscript ``$\sigma$.'' Let us consider the first situation that $g_{0,\varepsilon}+g_{2,\varepsilon}=0$ holds for all but two values of $\varepsilon$.  
Clearly, the equation $g_{0,\varepsilon}+g_{2,\varepsilon}=0$ is equivalent to  
\[
\varepsilon^2(x^2+y^2+z^2+w^2) + 2\varepsilon(ax+by+cz+dw) + a^2+b^2+c^2+d^2=0.
\]
Since at least three different values of  $\varepsilon$ 
satisfy the above equation, the coefficient of each term $\varepsilon^i$ ($i\in\{0,1,2\}$) should be zero. Therefore, the following Eq.(\ref{SymL-x2-a2-ax}) should hold. Eq.(\ref{SymL-x2-a2-ax}) also holds for $\sigma=(x_1x_3x_2)$ because an exchange of the two variables $x_2$ and $x_3$ does not change those equations. 
\begin{equation}\label{SymL-x2-a2-ax}
(x_1x_2x_3)  \;\text{or}\; (x_1x_3x_2) \;\;\;  x^2+y^2+z^2+w^2 = a^2+b^2+c^2+d^2 = ax+by+cz+dw =0.
\end{equation}
By permuting variable indices further, we obtain two more properties:
\begin{equation}\label{SymL-a2-c2-ac}
(x_2x_1x_3)  \;\text{or}\; (x_2x_3x_1) \;\;\; a^2+b^2+x^2+y^2 = c^2+d^2+z^2+w^2 = ac+bd+xz+yw =0.
\end{equation}\vs{-5}
\begin{equation}\label{SymL-a2-b2-ab}
(x_3x_2x_1)  \;\text{or}\; (x_3x_1x_2) \;\;\; a^2+c^2+x^2+z^2 = b^2+d^2+y^2+w^2 = ab+cd+xy+zw =0.
\end{equation}
For a later convenience, we claim that if all the above properties  hold then Eq.(\ref{a2+b2-relation})--(\ref{b2+d2-relation}) described below hold. This claim is proven as follows. 
{}From  $a^2+b^2+c^2+d^2 = c^2+d^2+z^2+w^2=0$ (Eq.(\ref{SymL-x2-a2-ax})--(\ref{SymL-a2-c2-ac})), we obtain  $a^2+b^2-z^2-w^2=0$. Similarly, from  $x^2+y^2+z^2+w^2=a^2+c^2+x^2+z^2=0$ (Eq.(\ref{SymL-x2-a2-ax})\&(\ref{SymL-a2-b2-ab})) follows  $a^2+c^2-y^2-w^2=0$. 
By combining these two obtained equations, we conclude that   $b^2+y^2-c^2-z^2=0$. 
Moreover,  $a^2+b^2+c^2+d^2=a^2+c^2+x^2+z^2=0$ (Eq.(\ref{SymL-x2-a2-ax})\&(\ref{SymL-a2-b2-ab})) implies $b^2+d^2-x^2-z^2=0$. {}From   $a^2+b^2+c^2+d^2= a^2+b^2+x^2+y^2= 0$ (Eq.(\ref{SymL-x2-a2-ax})--(\ref{SymL-a2-c2-ac})), we obtain  $c^2+d^2-x^2-y^2=0$, and   $a^2+b^2-z^2-w^2=b^2+d^2-x^2-z^2=0$ also implies $a^2+x^2-d^2-w^2=0$. 
In summary, we obtain two conditions given below. 
\begin{equation}\label{a2+b2-relation}
a^2+b^2-z^2-w^2 = a^2+c^2-y^2-w^2 
= b^2+y^2-c^2-z^2 =0.
\end{equation}
\begin{equation}\label{b2+d2-relation}
b^2+d^2-x^2-z^2 = c^2+d^2-x^2-y^2 
= a^2+x^2-d^2-w^2 = 0.
\end{equation}

We can further draw Eq.(\ref{a2+d2-relation}) by the following argument.  Assuming $b^2+d^2+y^2+w^2= a^2+b^2 -z^2-w^2 =0$ (Eq.(\ref{SymL-a2-b2-ab})\&(\ref{a2+b2-relation})), $a^2=d^2$ leads to $x^2=w^2$. Since its opposite direction holds as well, we conclude that  $a^2=d^2$ iff $x^2=w^2$. In a similar way, we obtain three more equivalence relations: $a^2=z^2$ iff $b^2=w^2$, $a^2=y^2$ iff $c^2=w^2$, and $b^2=c^2$ iff $y^2=z^2$. 
Overall, we can establish the following conditions. 
\begin{equation}\label{a2+d2-relation}
a^2=d^2 \IFF x^2=w^2, \;\; a^2=z^2 \IFF b^2=w^2, \;\; a^2=y^2 \IFF c^2=w^2.
\end{equation}

Next, let us recall  $xy+zw=-(ab+cd)$ (Eq.(\ref{SymL-a2-b2-ab})) and  $xz+yw=-(ac+bd)$ (Eq.(\ref{SymL-a2-c2-ac})). Using these equations, 
we can transform $(x+w)(y+z)$ into $-(a+d)(b+c)$ as follow.   
\[
(x+w)(y+z) = (xy+zw) + (xz+yw) = -(ab+cd) -(ac+bd) = -(a+d)(b+c).
\]
Thus, we immediately obtain the following equation. 
\begin{equation}\label{xwyz-adbc}
(x_1x_2x_3) \;\;\; (a+d)(b+c) + (x+w)(y+z) =0.
\end{equation}
By permuting variable indices, we also obtain the two more equations shown below.
\begin{equation}\label{aybx-cwdz}
(x_2x_1x_3) \;\;\; (a+y)(b+x) + (c+w)(d+z) =0.
\end{equation}
\begin{equation}\label{azcx-bwdy}
(x_3x_2x_1) \;\;\; (a+z)(c+x) + (b+w)(d+y) =0.
\end{equation}

\subsubsection{Situation 2: $g_0=g_2\wedge g_1=0$}\label{sec:SymL-sit-1}
  
Let us assume that both  $g_{0,\varepsilon}=g_{2,\varepsilon}$ 
and $g_{1,\varepsilon}=0$  hold  for at least three distinct values  of  $\varepsilon$. In what follows, we will discuss these two conditions separately. 

\s
{\bf [Case: $g_0=g_2$]} 
Consider the first case where $g_{0,\varepsilon}=g_{2,\varepsilon}$ holds  for at least three distinct values  of  $\varepsilon$. Using the value $[g_{0,\varepsilon},g_{1,\varepsilon},g_{2,\varepsilon}]$ 
given in Section \ref{sec:first-symmetrization}, the equation $g_{0,\varepsilon}-g_{2,\varepsilon}=0$ is equivalent to  
\[
\varepsilon^2(x^2+y^2-z^2-w^2) + 2\varepsilon(ax+by-cz-dw) + a^2+b^2-c^2-d^2=0.
\] 
Since there are three distinct values $\varepsilon$ satisfying the above equation, it follows that  
\begin{equation}\label{123-x2+y2-z2-w2}
(x_1x_2x_3) \;\;\; x^2+y^2-z^2-w^2 = a^2+b^2-c^2-d^2 = ax+by-cz-dw = 0
\end{equation}
Permuting variable indices further produces the following five more conditions. 
\begin{equation}\label{132-x2+z2-y2-w2}
(x_1x_3x_2) \;\;\; x^2+z^2-y^2-w^2 = a^2+c^2-b^2-d^2 = ax+cz-by-dw = 0.
\end{equation}
\begin{equation}\label{213-c2+d2-z2-w2}
(x_2x_1x_3) \;\;\; c^2+d^2-z^2-w^2 = a^2+b^2-x^2-y^2 = ac+bd-xz-yw = 0.
\end{equation}
\begin{equation}\label{231-c2+z2-d2-w2}
(x_2x_3x_1) \;\;\; c^2+z^2-d^2-w^2 = a^2+x^2-b^2-y^2 = ac+xz-bd-yw = 0.
\end{equation}
\begin{equation}\label{321-b2+y2-d2-w2} 
(x_3x_2x_1) \;\;\; b^2+y^2-d^2-w^2 = a^2+x^2-c^2-z^2 = ab+xy-cd-zw  = 0.
\end{equation}
\begin{equation}\label{312-b2+d2-y2-w2} 
(x_3x_1x_2) \;\;\; b^2+d^2-y^2-w^2 = a^2+c^2-x^2-z^2 = ab+cd-xy-zw = 0.
\end{equation}

Now, we claim, by the argument that follows, that Eq.(\ref{123-x2+y2-z2-w2})--(\ref{132-x2+z2-y2-w2}) imply 
$a^2=d^2$, $b^2=c^2$, $x^2=w^2$, $y^2=z^2$, $ax=dw$, and $by=cz$. 
{}From  $x^2+y^2=z^2+w^2$ (Eq.(\ref{123-x2+y2-z2-w2})) and $x^2+z^2=y^2+w^2$   (Eq.(\ref{132-x2+z2-y2-w2})) follows $y^2=z^2$; thus  $x^2=w^2$ also holds.  
Similarly, using both  $a^2+b^2=c^2+d^2$ (Eq.(\ref{123-x2+y2-z2-w2})) and  $a^2+c^2=b^2+d^2$  (Eq.(\ref{132-x2+z2-y2-w2})), we obtain  $b^2=c^2$ and  $a^2=d^2$. 
In addition, we obtain  $by=cz$ and $ax=dw$ from  $ax+by=cz+dw$ (Eq.(\ref{123-x2+y2-z2-w2})) and $ax+cz=by+dw$ (Eq.(\ref{132-x2+z2-y2-w2})).   Therefore, the claim should be true.

Similarly, Eq.(\ref{213-c2+d2-z2-w2})--(\ref{231-c2+z2-d2-w2}) imply that $a^2=z^2$, $b^2=w^2$, $c^2=x^2$, $d^2=y^2$, $ab=zw$, and $cd=xy$. 
Moreover, from Eq.(\ref{321-b2+y2-d2-w2})--(\ref{312-b2+d2-y2-w2}), it follows that $a^2=y^2$, $b^2=x^2$, $c^2=w^2$, $d^2=z^2$, $ac=yw$, 
and $bd=xz$. 

\s
{\bf [Case: $g_1=0$]} 
Let us consider the second case where $g_{1,\varepsilon}=0$ holds for at least three distinct values of $\varepsilon$. This case can be rephrased as    
\[
\varepsilon^2(xz+yw) + \varepsilon(az+bw+cx+dy) + ac+bd =0.
\]
Since this equation has degree at most $2$ with respect to the parameter $\varepsilon$, we can conclude the following. 
\begin{equation}\label{123-az+bw+cx+dy}
(x_1x_2x_3) \;\;\; az+bw+cx+dy = ac+bd = xz+yw =0.
\end{equation}
When permuting variable indices further, the following five conditions can be also induced. 
\begin{equation}\label{132-ay+cw+bx+dz}
(x_1x_3x_2) \;\;\; ay+bx+cw+dz = ab+cd = xy+zw =0.
\end{equation}
\begin{equation}\label{213-az+bw+cx+dy} 
(x_2x_1x_3) \;\;\; az+bw+cx+dy = ax+by = cz+dw =0.
\end{equation}
\begin{equation}\label{231-ad+bc+xw+yz} 
(x_2x_3x_1) \;\;\; ad+bc+xw+yz = ab+xy = cd+zw =0.
\end{equation}
\begin{equation}\label{321-ad+bc+xw+yz}
(x_3x_2x_1) \;\;\; ad+bc+xw+yz = ac+xz = bd+yw =0.
\end{equation}
\begin{equation}\label{312-az+bx+cw+dz} 
(x_3x_1x_2) \;\;\; ay+bx+cw+dz = ax+cz = by+dw =0.
\end{equation}

\subsubsection{Situation 3: $g_0g_2=g_1^{2}$}\label{sec:SymL-sit-3}

Let us consider the third situation that 
$g_{0,\varepsilon}^{\sigma}g_{2,\varepsilon}^{\sigma} = (g_{1,\varepsilon}^{\sigma})^2$ holds for at least five distinct values of  $\varepsilon$. This situation can be expressed as a degree-$4$ polynomial equation in $\varepsilon$. First, we fix $\sigma=(x_1x_2x_3)$ and omit superscript ``$\sigma$.''  Using the values  $g_{0,\varepsilon},g_{1,\varepsilon},g_{2,\varepsilon}$ given in Section \ref{sec:first-symmetrization}, the terms $g_{0,\varepsilon}g_{2,\varepsilon}$ and $(g_{1,\varepsilon})^2$ can be calculated as follows. 
\begin{eqnarray*}
g_{0,\varepsilon}g_{2,\varepsilon} &=& (x^2+y^2)(z^2+w^2)\varepsilon^4 + 2[(ax+by)(z^2+w^2) + (cz+dw)(x^2+y^2)]\varepsilon^3 \\
&& + [(x^2+y^2)(c^2+d^2) + (z^2+w^2)(a^2+b^2) + 4(ax+by)(cz+dw)] \varepsilon^2 \\
&& + 2[(ax+by)(c^2+d^2) +(cz+dw)(a^2+b^2)]\varepsilon + (a^2+b^2)(c^2+d^2).
\end{eqnarray*}\vs{-8}
\begin{eqnarray*}
(g_{1,\varepsilon})^2 &=& (xz+yw)^2\varepsilon^4 + 2(xz+yw)(az+bw+cx+dy)\varepsilon^3 \\
&& + [2(xz+yw)(ac+bd) +(az+bw+cx+dy)^2]\varepsilon^2 \\
&& + 2(ac+bd)(az+bw+cx+dy)\varepsilon + (ac+bd)^2. \hs{32}
\end{eqnarray*}
Since $g_{0,\varepsilon}g_{2,\varepsilon} = (g_{1,\varepsilon})^2$ holds for at least five distinct values of $\varepsilon$, 
coefficients of each term $\varepsilon^{i}$ ($i\in\{0,1,2,3\}$) in both $g_{0,\varepsilon}g_{2,\varepsilon}$ and $(g_{1,\varepsilon})^{2}$ coincide. For instance,  two coefficients of the term $\varepsilon^0$ in  $g_{0,\varepsilon}g_{2,\varepsilon}$ and $(g_{1,\varepsilon})^2$ are equal, and thus we obtain $(a^2+b^2)(c^2+d^2) = (xz+yw)^2$, which is 
equivalent to $ad=bc$. By a similar calculation of every term 
$\varepsilon^i$, the equation 
$g_{0,\varepsilon}g_{2,\varepsilon}=(g_{1,\varepsilon})^2$ implies 
the following.
\begin{equation}\label{123-ad-bc} 
(x_1x_2x_3) \;\;\; ad-bc = xw-yz = aw-bz -cy+dx =0.
\end{equation}
By permuting variable indices, we also obtain additional two 
sets of equations.
\begin{equation}\label{213-ay-bx} 
(x_2x_1x_3) \;\;\; ay-bx = cw-dz = aw -bz+cy -dx =0.
\end{equation}
\begin{equation}\label{321-az-cx} 
(x_3x_2x_1) \;\;\; az-cx = bw-dy = aw+bz -cy -dx=0. 
\end{equation}

\subsection{Basic Properties of Sym(f)}\label{sec:property-h}

Finally, we will present a set of basic properties concerning the symmetrization $Sym(f)$, where  
 $f=(a,b,c,d,x,y,z,w)$ is any ternary signature. Here, 
we fix  $\sigma \in\{ (x_1x_2x_3),(x_1x_3x_2)\}$.  
Each element of $Sym(f)=[h_0,h_1,h_2,h_3]$ can be calculated as follows. 
\begin{equation}\label{h-123-zero}\vs{-1}
h_{0} = (a+d)[(a+d)^2 +3(bc-ad)]. 
\end{equation}
\begin{equation}\label{h-123-one}
h_{1} = (a^2+bc)x + (a+d)(bz+cy) + (bc+d^2)w. 
\end{equation}
\begin{equation}\label{h-123-two}
h_{2} = a(x^2+yz) + (bz+cy)(x+w) + d(yz+w^2). 
\end{equation}
\begin{equation}\label{h-123-three}
h_{3} = (x+w)[(x+w)^2 +3(yz-xw)].
\end{equation}

\section{Proof of Proposition \ref{SymL-non-degenerate}}\label{sec:proof-SymL}

As promised in Section \ref{sec:first-symmetrization}, we will present the proof of Proposition \ref{SymL-non-degenerate}.  Our argument that will follow shortly is quite elementary and it requires only a straightforward analysis of a set of low-degree polynomial equations listed in Section \ref{sec:SymL-sit-3}.  An underlying goal of the analysis is to prove that such a set of equations has no common solution. 

Let $f=(a,b,c,d,x,y,z,w)$ denote an arbitrary ternary signature and assume 
that $f\not\in\DUP$. 
In addition, we denote by $\sigma$ an arbitrary permutation in $S_3$ and we set 
$SymL(f_{\sigma})=[g_{0,\varepsilon}^{\sigma},g_{1,\varepsilon}^{\sigma},g_{2,\varepsilon}^{\sigma}]$. 
To lead to a contradiction, we first assume that $SymL(f_{\sigma})$ is degenerate. More precisely, we assume that 
$g_{0,\varepsilon}^{\sigma}g_{2,\varepsilon}^{\sigma} =(g_{1,\varepsilon}^{\sigma})^2$ for at least five distinct values of  $\varepsilon$. As discussed in Section \ref{sec:SymL-sit-3}, this assumption implies Eq.(\ref{123-ad-bc})--(\ref{321-az-cx}). We split the proof into three situations, depending on the choice of $\sigma$. Since the third situation, in which $\sigma=(x_3x_2x_1)$ or $(x_3x_1x_2)$, is essentially the same as the first two situations, for readability, we omit this  situation. 
At last,  we conveniently set $\sigma_1=(x_1x_2x_3)$, $\sigma_2=(x_2x_1x_3)$, 
and $\sigma_3=(x_3x_2x_1)$.

\subsection{Situation: $\sigma=(x_1x_2x_3)$ or $(x_1x_3x_2)$}

Here, we consider only the situation where $\sigma=(x_1x_2x_3)$. 
For this $\sigma$, Eq.(\ref{123-ad-bc}) must hold; that is, $ad=bc$, $xw=yz$, and $aw+dx=bz+cy$. 
In what follows, we intend to show that $f$ belongs to $\DUP$ using Eq.(\ref{123-ad-bc}), because this clearly contradicts our assumption of $f\not\in\DUP$.  

\s
{\bf [Case: $ax\neq0$]}
Initially, we set $\gamma=\frac{b}{a}$ and $\delta  = \frac{y}{x}$ . {}From $ad=bc$  and  $xw=yz$, 
we obtain  $b=\gamma a$, $d=\gamma c$, $y=\delta x$, and $w=\delta z$. At this point, $f$ is expressed as $(a,\gamma a,c,\gamma c,x,\delta x,z,\delta z)$. 
{}From  $aw+dx=bz+cy$, it easily follows that  (1') $(\delta-\gamma)(az-cx)=0$; thus,  either $\delta=\gamma$ or $az=cx$ holds. Now, we discuss these two cases separately. 
When $\delta=\gamma$, $f_{\sigma_3}$ equals $[1,\gamma](x_3)\cdot (a,x,c,z,a,x,c,z)$; thus, $f$ belongs to $\DUP$.   
If $\delta\neq\gamma$, then (1') implies $az=cx$. Next, let $\theta=\frac{c}{a}$, implying $c=\theta a$ and $z=\theta x$ from $az=cx$. Since $d=\gamma c=\theta\gamma a$ and $w=\delta z=\theta\delta x$,  $f_{\sigma_2}$ becomes $[1,\theta](x_2)\cdot (a,\gamma a,x,\delta x,a,\gamma a,x,\delta x)$. This proves $f$ to be in $\DUP$. 

\s
{\bf [Case: $ax=0$]} 
Since this case is more involved, we split it into three subcases. 

\s
{[Subcase: $a=x=0$]} 
{}From $ad=bc$, we immediately obtain (3') 
$bc=0$, which implies either $b=0$ or $c=0$. Similarly, 
$xw=yz$ implies (4') $yz=0$, which means either $y=0$ or $z=0$.
Firstly, we assume that $b=y=0$. For the permutation $\sigma_2$, this assumption makes $f_{\sigma_2}$ equal $(0,0,0,0,c,d,z,w)$,  and thus $f$  belongs to $\DUP$. 
Secondly, we assume that $b=0\wedge y\neq0$. {}From (4') follows $z=0$.   
By $aw+dx=bz+cy$, we obtain $cy=0$, which yields $c=0$. 
For  $\sigma_3$, $f_{\sigma_3}$ becomes $(0,0,0,0,0,y,d,w)$, again in 
$\DUP$. 
Thirdly, we consider the case where $b\neq0\wedge y=0$. 
Using (3'), we deduce $c=0$. 
{}From $aw+dx=bz+cy$, we also obtain $bz=0$, implying $z=0$. 
Since $f_{\sigma_3} = (0,0,0,0,b,y,d,w)$, obviously $f$ belongs to $\DUP$. 
Finally, we discuss the case where  $b\neq0\wedge y\neq0$. 
The two equations (3') and (4') indicate that $c=z=0$. Moreover, 
we  obtain $f_{\sigma_3} = [1,\gamma](x_3)\cdot (0,x,0,0,0,x,0,0)$, 
making $f$ fall into $\DUP$.   In all the cases, contradictions follow. 

\s
{[Subcase: $a=0\wedge x\neq0$]} 
{}From $ad=bc$, we have (5') 
$bc=0$, which implies either $b=0$ or $c=0$. 
setting $\gamma=\frac{y}{x}$, we obtain $y=\gamma x$ and $w=\gamma z$ from  $xw=yz$. 
Now, we begin with examining the case of $b=0$. Since $aw+dx=bz+cy$,  it holds that $x(d-\gamma c)=0$; thus, $d=\gamma c$ follows. This concludes that $f_{\sigma_3} = [1,\gamma](x_3)\cdot (0,x,c,z,0,x,c,z)$. Obviously, this makes $f$ fall into $\DUP$. 
Next, let us consider the case of $b\neq0$. {}From (5') follows $c=0$. 
We also obtain $dx=bz$ from $aw+dx=bz+cy$. Letting $\delta =\frac{z}{x}$, we further obtain $z=\delta x$ and $d=\delta b$ from $dx=bz$. Note that 
$xw=yz$ implies $\gamma x(z-\delta x)=0$, yielding $z=\delta x$.  It thus  holds that $w=\gamma z= \delta \gamma z$. For the permutation  $\sigma_2$, $f_{\sigma_2}$ can be written in the form $[1,\delta](x_2)\cdot (0,b,x,\gamma x,0,b,x,\gamma x)$, which is clearly in $\DUP$. 

\s
{[Subcase: $a\neq0\wedge x=0$]} 
Because this subcase is essentially the same as the previous subcase $a=0\wedge x\neq0$, we omit this subcase for readability. 

\subsection{Situation: $\sigma=(x_2x_1x_3)$ or $(x_2x_3x_1)$}

In this subsection, we assume that $\sigma=(x_2x_1x_3)$. Notice that our assumption $g_0^{\sigma}g_2^{\sigma}=(g_1^{\sigma})^2$ 
ensures Eq.(\ref{213-ay-bx}); that is, $ay=bx$, $cw=dz$, 
and $aw+cy=bz+dx$.  With these equations, we wish to lead to a contradiction. 

\s
{\bf [Case: $az\neq0$]}
Using $ay=bx$ and $cw=dz$, we conveniently set $\gamma =\frac{b}{a}$ and $\delta =\frac{w}{z}$; thus, $\gamma$ and $\delta$ satisfy that 
$b=\gamma a$, $d=\delta c$, $y=\gamma x$, and $w=\delta z$. 
{}From $aw+cy=bz+dx$, it follows that (1') $(\delta-\gamma)(az-cx)=0$. 
Hereafter, let us consider two subcases: $\delta=\gamma$ and $\delta\neq\gamma$.
First, we assume that $\delta=\gamma$. Obviously, $f_{\sigma_3}$ equals  $[1,\gamma](x_3)\cdot (a,x,c,z,a,x,c,z)$, and thus $f$ belongs to $\DUP$. 
Next, we assume that $\delta\neq\gamma$. 
Clearly, (1') implies $az=cx$. Note that $c\neq0$ because of $az\neq0$. Now, let $\theta =\frac{c}{a}$; thus, $c=\theta a$ and $z=\theta x$ hold. Using this $\theta$,  $f$ can be expressed as $[a,x](x_1)\cdot (1,\gamma,\theta,\theta\delta,1,\gamma,\theta,\theta\delta)$, which is clearly in $\DUP$. 

\s
{\bf [Case: $az=0$]}
To handle this case, we will consider three subcases.

\s
{[Subcase: $a=z=0$]} 
By $ay=bx$, we obtain (2') $bx=0$, implying either $x=0$ or $b=0$.  Similarly,  
$cw=dz$ implies (3') $cw=0$; thus, either $c=0$ or $w=0$ holds. 
Firstly, we assume that $c=x=0$. This implies that $f_{\sigma_3}$ is of 
the form  $(0,0,0,0,b,y,d,w)$, which forces $f$ to be in $\DUP$. 
Secondly, we assume that $c=0\wedge x\neq0$. {}From (2') follows $b=0$. 
Since $x\neq0$, we obtain $d=0$ from $aw+cy=bz+dx$. Therefore, it holds that $f = (0,0,0,0,x,y,z,\gamma z)$, proving that $f\in\DUP$. 
Thirdly, we assume that $c\neq0\wedge x=0$. Note that $w=0$ by (3'). 
The equation $aw+cy=bz+cy$ yields $c=0$; hence, $f_{\sigma_3}$ becomes $(0,0,0,0,b,y,d,0)\in\DUP$. 
The remaining case is that $c\neq0\wedge x\neq0$. {}From (2')\&(3') follows $b=w=0$.  
The equation $aw+cy=bz+dx$ is thus equivalent to $dx=cy$. 
If we set $\gamma =\frac{c}{d}$, then we obtain $c=\gamma d$ and $x=\gamma y$ from $dx=cy$, and thus $f_{\sigma_3}$ can be written as $[\gamma,1](x_3)\cdot (0,y,d,0,0,y,d,0)$. Clearly, $f$ belongs to $\DUP$. 

\s
{[Subcase: $a=0\wedge z\neq0$]}  
{}From $ay=bx$, we obtain (4') $bx=0$. Letting  $\gamma =\frac{w}{z}$, 
we obtain $w=\gamma z$ and $d=\gamma c$ from  $cw=dz$. 
Firstly, we assume that $b=c=0$; thus, $d=\gamma c=0$. We immediately obtain $f=(0,0,0,0,x,y,z,\gamma z)\in\DUP$. 
Secondly, assume that $b=0\wedge c\neq0$. Since $aw+cy=bz+dx$ is equivalent to $c(y-\gamma x)=0$, $c\neq0$ implies $y=\gamma x$. Thus, $f_{\sigma_3}$ becomes $[1,\gamma](x_3)\cdot (0,x,c,z,0,x,c,z)$. This implies that $f\in\DUP$. 
Finally, let us handle the case of $b\neq0$. Here, we obtain $x=0$ by (4'). Using $aw+cy=bz+dx$, we also obtain $cy=bz$. 
Now, let $\delta=\frac{y}{b}$ since $b\neq0$. With this $\delta$, 
it follows that $y=\delta b$ and $z=\delta c$. Obviously, $f$ equals $[1,\delta](x_1)\cdot (0,b,c,\gamma c,0,b,c,\gamma c)$. Obviously, $f$ 
belongs to $\DUP$. 

\s
{[Subcase: $a\neq0\wedge z=0$]}   
Note that (5') $cw=0$ is obtained from $cw=dz$. Now, let $\gamma=\frac{b}{a}$; thus, $ay=bx$ implies both  $b=\gamma a$ and $y=\gamma x$. 
First of all, we consider the case where $c=0$. Note that $aw+cy=bz+dx$ immediately  leads to  $aw=dx$.  Conveniently, we set $\delta =\frac{d}{a}$. 
It then follows  from $aw=dx$  that $d=\delta a$ and $w=\delta x$. 
Hence, we obtain $f=[a,x](x_1)\cdot (1,\gamma,0,\delta,1,\gamma,0,\delta)\in\DUP$.  
What still remains is the case where $c\neq0$. By (5'), we immediately  obtain $w=0$. Moreover,  $aw+cy=bz+dx$ implies $x(d-\gamma c)=0$. If $x\neq0$, then $d=\gamma c$ also follows. In summary, $f_{\sigma_3}$ must have the form  $[1,\gamma](x_3)\cdot (a,x,c,0,a,x,c,0)$, proving that $f\in\DUP$.  
On the contrary, if $x=0$, then we immediately obtain $f=(a,\gamma a,c,\gamma c,0,0,0,0)$. This makes $f$ fall into $\DUP$, as requested.

\section{Proof of Proposition \ref{prop-type-I}}\label{sec:type-I} 

Here, we will prove Proposition \ref{prop-type-I}.  In this proof, 
we assume that $f$ is of the form $(a,b,c,d,x,y,z,w)$ and let $SymL(f_{\sigma})_{\varepsilon} =[g_{0,\varepsilon}^{\sigma},g_{1,\varepsilon}^{\sigma},g_{2,\varepsilon}^{\sigma}]$ for each permutation $\sigma$ and each value $\varepsilon$. 
Furthermore, we assume that $f$ is $SIG_1$-legal; that is, the signature  $Sym(f)=[h_0,h_1,h_2,h_3]$ satisfies $h_0+h_2=h_1+h_3=0$ and $h_0\neq \xi h_1$ for any constant $\xi\in\{\pm i\}$. Toward a contradiction, we further 
assume that, for every permutation $\sigma$ and almost all values  of $\varepsilon$, both $g_{0,\varepsilon}^{\sigma} = g_{2,\varepsilon}^{\sigma}$ and $g_{1,\varepsilon}^{\sigma} =0$ hold.
Notice that this assumption implies Eq.(\ref{123-x2+y2-z2-w2})--(\ref{312-az+bx+cw+dz}).  
As shown in Section \ref{sec:SymL-sit-1}, Eq.(\ref{123-x2+y2-z2-w2})--(\ref{312-b2+d2-y2-w2}) imply that  $a^2=d^2=y^2=z^2$ and $b^2=c^2=x^2=w^2$. 
{}From these equations, we can set $z=e_1a$, $y=e_2a$, $d=e_3a$, $b=e_4w$, $c=e_5w$, and $x=e_6w$ using appropriate constants $e_i\in\{\pm1\}$. 
Eq.(\ref{123-x2+y2-z2-w2})--(\ref{312-b2+d2-y2-w2}) also provide  with the following equations: $ax=dw$, $by=cz$, $ac=yw$, $bd=xz$, $ab=zw$, and $cd=xy$. Now, we split our proof into two cases, depending on whether $aw=0$ or not, and we try to argue  that each case indeed leads to a contradiction. 

\s
{\bf [Case: $aw\neq0$]}   
{}From $ax=dw$, we obtain $e_6 aw= e_3 aw$, or equivalently $(e_6-e_3)aw=0$; thus, $e_3=e_6$ must hold since $aw\neq0$. Similarly, from $ac=yw$ and $ab=zw$, it follows that $e_1=e_4$ and $e_2=e_5$, respectively.  
Moreover, $ac+bd=0$ (Eq.(\ref{123-az+bw+cx+dy})) implies 
 $(e_2+e_1e_3)aw=0$, which yields $e_3=-e_1e_2$. 
Similarly, from $az+bw+cx+dy=0$  (Eq.(\ref{123-az+bw+cx+dy})) follows  $2e_1(a^2+w^2)=0$; hence, we obtain $a^2+w^2=0$. Let us assume that  $w=\gamma a$ for an appropriate constant $\gamma\in\{\pm i\}$. At present,  $f$ equals  
$(a,e_1\gamma a,e_2\gamma a,-e_1e_2 a,-e_1e_2\gamma a,e_2 a,e_1a,\gamma a)$. 
Next, let us consider the values $h_0$ and $h_1$.   
Making a direct calculation of Eq.(\ref{h-123-zero})--(\ref{h-123-one}), we obtain $h_0 = (1-e_1e_2)^3a^3$ and $h_1 = \gamma (1-e_1e_2)(3-e_1e_2)a^3$. When $e_1e_2=1$, it clearly follows that  
$h_0 = h_1=0$, a contradiction against $h_0\neq\xi h_1$ for every $\xi\in\{\pm i\}$; therefore, $e_1e_2$ must be $-1$, or equivalently $e_2=-e_1$. Using this result, we further simplify $h_0$ 
and $h_1$ as $h_0 =8a^3$ and $h_1= 8\gamma a^3$. These values imply  $h_1=\gamma h_0$. Since $\gamma\in\{\pm i\}$, this equality leads to a contradiction, as requested. 

\s
{\bf [Case: $aw=0$]} 
First, note that both $a=0$ and $w=0$ never happen simultaneously because, otherwise, $f$ becomes an all-zero function, and thus $f$ belongs to $\DUP$, a contradiction. When $a=0$, $f$ equals $(0,e_1w,e_2w,0,e_3w,0,0,w)$. 
{}From  $az+bw+cx+dy=0$ (Eq.(\ref{123-az+bw+cx+dy})) follows  $(e_1+e_2e_3)w^2=0$, which implies $e_3=-e_1e_2$. Hence, we obtain 
$f=w\cdot (0,e_1,e_2,0,-e_1e_2,0,0,1)$. 
By Eq.(\ref{h-123-zero})--(\ref{h-123-one}), it follows that 
$h_1 = e_1e_2-1$ and $h_3=2+e_1e_2$; as a result, $h_1+h_3= 1+2e_1e_2\neq 0$ follows. 
This consequence clearly contradicts the assumption that $h_1+h_3=0$.  
Similarly, when $w=0$, since $a\neq0$, $f$ equals $(a,0,0,e_3 a,0,e_2 a,e_1 a,0)$. Using  $az+bw+cx+dy=0$ (Eq.(\ref{123-az+bw+cx+dy})), we obtain $(e_1+e_2e_3)a^2=0$, implying $e_3=-e_1e_2$. This makes $f$ equal   
$a\cdot (1,0,0,-e_1e_2,0,e_2,e_1,0)$.  Since $h_0 =2+e_1e_2$ and $h_2=e_1e_2-1$, we then conclude that $h_0+h_2 =1+2e_1e_2\neq 0$, a contradiction against $h_0+h_2=0$. 

\section{Proof of Proposition \ref{prop-type-II}}\label{sec:type-II}

Assume that $f=(a,b,c,d,x,y,z,w)\in\DUP$ is $SIG_1$-legal  and let $SymL(f_{\sigma})=[g^{\sigma}_0,g^{\sigma}_1,g^{\sigma}_2]$ for any permutation $\sigma\in S_3$.  
Here, we aim at proving Proposition \ref{prop-type-II} by contradiction. To achieve this goal, we first assume that, together with Statement (*),  there are two distinct permutations $\sigma$ and $\tau$ for which (i)  $g_0^{\sigma}=g_2^{\sigma} \wedge g_1^{\sigma}=0$ and (ii) $g_0^{\tau}+g_2^{\tau}=0$ hold. 
{}From this assumption, we want to lead to a contradiction. 
As shown in Section \ref{sec:outline-proof}, Statement (*) implies that,  for every $\sigma'\in S_3$, the two conditions (i) and (ii) are not satisfied simultaneously. 
Since $f$ is $SIG_1$-legal, it also holds that $h_0+h_2=h_1+h_3=0$ and $h_0^2+h_1^2\neq0$, provided that $Sym(f)=[h_0,h_1,h_2,h_3]$. Notice that $h_2^2+h_3^2\neq0$ also holds. 

\subsection{Situation: $\sigma=(x_1x_2x_3)$ and $\tau=(x_2x_1x_3)$}\label{sec:situation-01}

For our choice of $\sigma$ and $\tau$, we assume that $g_0^{\sigma}=g_2^{\sigma}\wedge g_1^{\sigma}=0$  and  $g_0^{\tau}+g_2^{\tau}=0$. Letting $\sigma'=(x_1x_3x_2)$,  
we first claim that $g_0^{\sigma'}+g_2^{\sigma'}\neq0$ holds. Meanwhile, assume  otherwise. Because of the close similarity between $\sigma$ and $\sigma'$,  as seen in Section \ref{sec:SymL-sit-2}, $g_0^{\sigma}+g_2^{\sigma}=0$ should hold for $\sigma$. This indicates the condition $g_0^{\sigma}=g_2^{\sigma}\wedge g_1^{\sigma}=0$ to fail; thus, we obtain a contradiction. Therefore, since $g_0^{\sigma'}+g_2^{\sigma'}\neq0$, we conclude that   $g_0^{\sigma'}=g_2^{\sigma'}\wedge g_1^{\sigma'}=0$. 

{}From our assumption, Eq.(\ref{123-x2+y2-z2-w2})--(\ref{132-x2+z2-y2-w2}) and Eq.(\ref{123-az+bw+cx+dy})--(\ref{132-ay+cw+bx+dz}) hold  respectively 
for $\sigma$ and $\sigma'$, and Eq.(\ref{SymL-a2-c2-ac}) 
holds for $\tau$.  
As Section \ref{sec:SymL-sit-1} showed,  
Eq.(\ref{123-x2+y2-z2-w2})--(\ref{132-x2+z2-y2-w2}) produce the following six simple equations:  $a^2=d^2$, $b^2=c^2$, $x^2=w^2$, 
$y^2=z^2$, (1') $ax=dw$, and (2') $by=cz$. 
Since $a^2=d^2$, we assume that $d=e_1a$ for a certain constant 
$e_1\in\{\pm 1\}$. Similarly, using three relations, $b^2=c^2$, $x^2=w^2$, and $y^2=z^2$, it is possible to  
set $c=e_2b$, $w=e_3x$ and $z=e_4y$ using appropriate constants 
$e_2,e_3,e_4\in\{\pm 1\}$. Let us examine the following two cases. 

\s
{\bf [Case: $a=0$]}
We split this case into two subcases, depending on whether $x=0$ or not.  The first subcase is rather simple. 
Note that $d=0$ holds because $d=e_1a$.    

\s
[Subcase: $x=0$]  
Clearly,  $w=e_3x=0$ holds. We also obtain $b^2+y^2=0$ because  $a^2+b^2+x^2+y^2=0$ (Eq.(\ref{SymL-a2-c2-ac})) holds.  
{}From this equation, we conclude that $b=0$ iff $y=0$. In particular, if $by=0$, then $f$ is composed of all zeros, forcing $f$ fall into $\DUP$, a contradiction. It thus suffices to assume that $by\neq0$. By (2'), we obtain  $(1-e_2e_4)by=0$; thus, $e_2e_4=1$, or equivalently, $e_4=e_2$ holds. A vigorous calculation of Eq.(\ref{h-123-zero})--(\ref{h-123-one})  shows that $h_0 = h_1=0$. This is a contradiction against our requirement that $h_1\neq \xi h_0$ for any $\xi\in\{\pm i\}$. 

\s
[Subcase: $x\neq0$] 
First, we want to claim that $b\neq0$. Assume otherwise.  Since  
$b=0$ implies $c=e_2b=0$, it follows that $a=b=c=d=0$. We therefore conclude that $f$ is in $\DUP$. This is a clear contradiction; therefore,  
$b\neq0$ should hold. Using Eq.(\ref{h-123-zero})--(\ref{h-123-two}), we obtain  $h_0=0$, $h_1= (1+e_3)e_2b^2x$, and $h_2=(e_2+e_4)(1+e_3)bxy$.    Since $h_0\neq \xi h_1$ for any $\xi\in\{\pm i\}$,  $h_1\neq0$ must hold; thus, $e_3\neq-1$, or equivalently $e_3=1$ follows.  Therefore, $h_2$ is of the form $h_2 = 2(e_2+e_4)bxy$. 
First, let us consider the case where  $y\neq0$. Since $h_0+h_2=0$, we obtain  
$2(e_2+e_4)bxy=0$, which yields $e_4=-e_2$. By contrast, from  (2') follows  $(1-e_2e_4)by=0$. We thus conclude 
that $e_2e_4=1$, or equivalently $e_4=e_2$. This is obviously a contradiction. 
Next, consider the case where $y=0$. We can simplify  $az+bw+cx+dy=0$ (Eq.(\ref{123-az+bw+cx+dy})) to  $(1+e_2)bx=0$; thus, $e_2=-1$ follows.  Similarly, from $a^2+b^2+x^2+y^2=0$ (Eq.(\ref{SymL-a2-c2-ac})), we deduce  (3') $b^2+x^2=0$.  The values $h_1$ and $h_3$ take  $h_1 = -2b^2x$ and 
$h_3 = 2x^3$ by Eq.(\ref{h-123-one})\&(\ref{h-123-three}). The requirement $h_1+h_3=0$ implies $2x(x^2-b^2)=0$; thus, $x^2=b^2$ follows. By combining this equation with (3'), we conclude that $x=b=0$. This is obviously a contradiction against $b\neq0$. 

\s
{\bf [Case: $a\neq0$]}
This case is more involved. 
Similar to the previous case, we split this case into two subcases.

\s
[Subcase: $x=0$] Note that $w=e_3x=0$. 

(i) We start with assuming $by\neq0$. Using  $az+bw+cx+dy=0$ (Eq.(\ref{123-az+bw+cx+dy})), we deduce $(e_1+e_4)ay=0$, from which $e_4=-e_1$ follows. Similarly, from  $ac+bd=0$ (Eq.(\ref{123-az+bw+cx+dy})), we obtain  $(e_1+e_2)ay=0$ and then $e_2=-e_1$. Now, let us determine the value $e_1$  using Eq.(\ref{h-123-zero})--(\ref{h-123-three}). Since $h_3=0$ and $h_1=-2e_1(1+e_1)aby$ by a direct calculation, the requirement $h_1+h_3=0$  leads to  $e_1(1+e_1)aby=0$, further implying $e_1=-1$. At present, $f$ has the form $(a,b,b,-a,0,y,y,0)$. Since the value $h_2$ becomes $0$, we  therefore conclude that $h_2=h_3=0$, contradicting the requirement $h_1^2+h_3^2\neq0$.  

(ii) Next, we assume that $b=y=0$. 
Since  $a^2+b^2+x^2+y^2=0$ (Eq.(\ref{SymL-a2-c2-ac})), we immediately 
obtain $a=0$. This contradicts our assumption $a\neq0$. 

(iii) Let us assume that $b=0\wedge y\neq0$. Note that $c=e_2b=0$. 
The equation $az+bw+cx+dy=0$ (Eq.(\ref{123-az+bw+cx+dy})) implies  $(e_1+e_4)ay=0$, which yields $e_4=-e_1$. It thus follows 
by Eq.(\ref{h-123-two})--(\ref{h-123-three}) that $h_2 = -(1+e_1)ay^2$ and $h_3=0$.  Here, we claim that $e_1\neq-1$ because, otherwise,  
we obtain $h_2=h_3=0$, a contradiction. Since $e_1\neq-1$,  $e_1=1$ must hold. The value $h_2$ then becomes $h_2=-2ay^2$. Since $h_0=2a^3$, the requirement $h_0+h_2=0$ implies $2a(a^2-y^2)=0$, which is equivalent to (4') $a^2=y^2$. Next, we use   $a^2+b^2+x^2+y^2=0$ (Eq.(\ref{SymL-a2-c2-ac})) 
to obtain $a^2+y^2=0$. {}From (4'), we conclude that $a=y=0$. This is a clear contradiction. 

(iv) Finally, we assume that $b\neq0\wedge y=0$. Obviously,  $z=e_4y=0$ holds. We then obtain $(e_1+e_2)ab=0$ from  $ac+bd=0$ (Eq.(\ref{123-az+bw+cx+dy})).  This yields $e_2=-e_1$. By a simple calculation, we obtain  $h_2=h_3=0$, from which a contradiction follows. 

\s
[Subcase: $x\neq0$] 
We use (1') to obtain $(1-e_2e_3)ax=0$, from which we conclude that $e_2e_3=1$, or equivalently $e_2=e_3$. 

(i) Assume that $by\neq0$. It follows from (2') that  $(1-e_2e_4)by=0$;, thus,  $e_4=e_2$ holds. Because of  $xz+yw=0$ (Eq.(\ref{123-az+bw+cx+dy})), we conclude that $2e_2xy=0$. This implies that either $x=0$ or $y=0$, and it clearly contradicts our current assumption. 

(ii) Assuming that $b=y=0$,  we can simplify  $ax+by-cz-dw=0$ (Eq.(\ref{123-x2+y2-z2-w2})) to $(1-e_1e_2)ax=0$, further implying $e_2=e_1$. Now, we show that $e_1=1$. For this purpose, we first calculate $h_2$ and $h_3$ as  $h_2=(1+e_1)ax^2$ and $h_3 = (1+e_1)(2-e_1)x^3$. If $e_1=-1$, then $h_2=h_3=0$ follows. Since this is a contradiction, it must hold that $e_1\neq-1$, or equivalently $e_1=1$, as requested. 
The equation $a^2+b^2+x^2+y^2=0$ (Eq.(\ref{SymL-a2-c2-ac})) then becomes   $a^2+x^2=0$. Now, we set $x=\gamma a$ using an appropriate constant $\gamma\in\{\pm i\}$. It is easy to show that $h_0 = 2a^3$ and $h_1 =2\gamma a^3$; thus,  $h_1 =\gamma h_0$ holds, a contradiction. 

(iii) Next, we assume that $b=0\wedge y\neq0$. It follows from   $xz+yw=0$ (Eq.(\ref{123-az+bw+cx+dy}))  that $(e_2+e_4)xy=0$; thus, $e_4=-e_2$ holds.  By  $az+bw+cx+dy=0$ (Eq.(\ref{123-az+bw+cx+dy})), we also obtain $(e_1-e_2)ay=0$, from which $e_2=e_1$ follows.  
Now, we want to claim that $e_1=1$. This is shown as follows. Note that $h_0 = (1+e_1)(2-e_1)a^3$ and $h_1 = (1+e_1)a^2x$. If $e_1=-1$, then we immediately obtain $h_0=h_1=0$, contradicting the requirement $h_0^2 +h_1^2\neq0$. Since $e_1\in\{\pm 1\}$, $e_1=1$ follows.
Therefore, it holds that $h_0 = 2a^3$ and $h_2 = 2a(x^2-y^2)$. Since $h_0+h_2=0$, we obtain $2a(a^2+x^2-y^2)=0$; thus, $a^2+x^2-y^2=0$ follows.  Now,  $a^2+b^2+x^2+y^2=0$ (Eq.(\ref{SymL-a2-c2-ac})) becomes  $a^2+x^2+y^2=0$. These two equations clearly imply $y=0$, a contradiction against $y\neq0$. 

(iv) The remaining case is that $b\neq0\wedge y=0$. By  $ac+bd=0$ (Eq.(\ref{123-az+bw+cx+dy})), it follows that $(e_1+e_2)ab=0$; thus, we have $e_2=-e_1$. Moreover, from  $ax+by-cz-dw=0$ (Eq.(\ref{123-x2+y2-z2-w2})) follows  $(1+e_1)ax=0$, yielding $e_1=-1$. The equation  $az+bw+cx+dy=0$ (Eq.(\ref{123-az+bw+cx+dy})) therefore becomes equivalent to $bx=0$, leading to 
a contradiction against $b\neq0$ and $x\neq0$. 

\subsection{Situation: $\sigma=(x_2x_1x_3)$ and $\tau=(x_1x_2x_3)$}

Let us assume that  $g_0^{\sigma}=g_2^{\sigma}\wedge g_1^{\sigma}=0$ for  $\sigma=(x_2x_1x_3)$ and $g_0^{\tau}+g_2^{\tau}=0$ for $\tau=(x_1x_2x_3)$.   For brevity, we set $\sigma'=(x_2x_3x_1)$ and $\sigma_3=(x_3x_2x_1)$.  
Following a similar argument given in Section \ref{sec:situation-01}, we can conclude another condition that $g_0^{\sigma'}=g_2^{\sigma'}\wedge g_1^{\sigma'}=0$ for 
$\sigma'$. Notice that our assumption guarantees Eq.(\ref{213-c2+d2-z2-w2})--(\ref{231-c2+z2-d2-w2}) and 
 Eq.(\ref{213-az+bw+cx+dy})--(\ref{231-ad+bc+xw+yz}) for $\sigma$ and $\sigma'$, respectively, and also Eq.(\ref{SymL-x2-a2-ax}) for $\tau$. 
As discussed in Section \ref{sec:SymL-sit-1}, Eq.(\ref{213-c2+d2-z2-w2})--(\ref{231-c2+z2-d2-w2}) implies the following equations: $a^2=z^2$, $b^2=w^2$, $c^2=x^2$, $d^2=y^2$, (1') $ab=zw$, and (2') $cd=xy$. 
With appropriate constants $e_1,e_2,e_3,e_4\in\{\pm 1\}$,  we can set   $z=e_1a$, $w=e_2b$, $x=e_3c$, and $y=e_4d$. 

\s
{\bf [Case: $a=0$]}
First, we obtain $z=0$ from $z=e_1a$. In what follows, we will discuss two subcases. 

\s
{[Subcase: $b=0$]} 
Since  $b=0$,  $w=0$ follows. Now, we claim that $e_4=e_3$. To show this claim, assume that $e_4\neq e_3$, or equivalently $e_3e_4\neq1$. {}From (2'), we obtain  $(1-e_3e_4)cd=0$, which means $cd=0$. The equation     $a^2+b^2+c^2+d^2=0$ (Eq.(\ref{SymL-x2-a2-ax}))  is then equivalent to  $c^2+d^2=0$. Moreover, $cd=0$ and $c^2+d^2=0$ imply $c=d=0$. Hence, $f$ is composed of all zeros, and thus it is in 
$\DUP$, a contradiction. As a consequence, we conclude that $e_4=e_3$. 
For $\sigma_3$, $f_{\sigma_3}$ becomes $(0,e_3c,c,0,0,e_3,d,d,0)$, which is written as $[c,d](x_3)\cdot (0,e_3,1,0,0,e_3,1,0)$. Thus, $f$ belongs to $\DUP$.  

\s
{[Subcase: $b\neq0$]}
There are two situations to consider separately. 

(i) Let us consider the case where $d=0$.  Note that $b^2=c^2$ follows  from  $a^2+x^2=b^2+y^2$ (Eq.(\ref{231-c2+z2-d2-w2})). Moreover, from  $a^2+b^2+c^2+d^2=0$ (Eq.(\ref{SymL-x2-a2-ax})), we conclude that 
 $b^2+c^2=0$. These two equations immediately yield $b=c=0$, which  contradicts $b\neq0$.  

(ii) Next, consider the case where $d\neq0$. Note that $(e_2+e_4)bd=0$ holds since $ax+by+cz+dw=0$ (Eq.(\ref{SymL-x2-a2-ax})); thus, $e_4=-e_2$ holds. 
Firstly, we assume that $c\neq0$. It follows by (2') that $(1+e_2e_3)cd=0$; hence, we obtain $e_3=-e_2$. {}From  $a^2+b^2=x^2+y^2$ (Eq.(\ref{213-c2+d2-z2-w2})) and  
$a^2+b^2+c^2+d^2=0$ (Eq.(\ref{SymL-x2-a2-ax})), it also follows that $b^2-c^2-d^2=0$ and $b^2+c^2+d^2=0$, respectively. Combining these two equations, we lead to $2b^2=0$, a contradiction. 
Secondly, we assume that $c=0$.  Note that $x=z=0$. The equation $a^2+b^2+c^2+d^2=0$ (Eq.(\ref{SymL-x2-a2-ax})) implies $b^2+d^2=0$. Furthermore, from $c^2+d^2-z^2-w^2=0$ (Eq.(\ref{213-c2+d2-z2-w2})), we obtain $b^2=d^2$. Combining these two consequences, we conclude that $b=d=0$. Hence, $f$ is an all-zero function and belongs to $\DUP$, a contradiction. 

\s
{\bf [Case: $a\neq0$]}
Here, we will consider two subcases. 

\s
{[Subcase: $bd\neq0$]} 
{}From (1'), we have $(1-e_1e_2)ab=0$. Thus, we have $e_2=e_1$. 

(i) Assume that $c=0$; thus, $x=e_3c=0$ holds. We deduce from  
$ax+by+cz+dw=0$ (Eq.(\ref{SymL-x2-a2-ax})) the equation $(e_1+e_4)bd=0$,  which leads to $e_4=-e_1$. Use  
$ac+bd=xz+yw$ (Eq.(\ref{213-c2+d2-z2-w2})), and we then obtain $2bd=0$; however, this is a contradiction against our assumption. 

(ii) Next, assume that $c\neq0$. The equation (2') implies $(1-e_3e_4)cd=0$, yielding $e_4=e_3$. {}From  
$ax+by+cz+dw=0$ (Eq.(\ref{SymL-x2-a2-ax})), it follows that (3') 
$(e_1+e_3)(ac+bd)=0$. This implies either $e_1+e_3=0$ or $ac+bd=0$. Here, we will examine these two possibilities. 

(a) Assume that $e_1+e_3=0$, or equivalently $e_3=-e_1$. {}From  $c^2+d^2=z^2+w^2$ (Eq.(\ref{213-c2+d2-z2-w2})), we obtain 
 $a^2+b^2 - c^2-d^2=0$. Combining this equation with  $a^2+b^2+c^2+d^2=0$ (Eq.(\ref{SymL-x2-a2-ax})), we also obtain $a^2+b^2=0$, from which  $c^2+d^2=0$ immediately follows. Now, we set $b=\gamma a$ and $d=\delta c$ for two constants $\gamma,\delta\in\{\pm i\}$. {}From  
$ac+bd=xz+yw$  (Eq.(\ref{213-c2+d2-z2-w2})), it follows 
that $2(1+\delta\gamma)ac=0$. Since $ac\neq0$, we conclude that 
$\gamma\delta=1$, or equivalently $\delta=\gamma$. Overall, $f_{\sigma_3}$  has the form $[1,\gamma](x_3)\cdot (a,-e_1c,c,e_1a,a,-e_1c,c,e_1a)$. Clearly, this contradicts $f\not\in\DUP$. 

(b) Assume that $e_1+e_3\neq0$; thus, $e_3\neq -e_1$, or equivalently  
 $e_3=e_1$ follows. By (3'), we obtain $ac+bd=0$.  Letting  $\gamma=\frac{b}{a}$, we obtain  $b=\gamma a$ and $c=-\gamma d$ from $ac+bd=0$.  Next, we claim that $\gamma^2=-1$. Assume otherwise. 
The equation   $a^2+b^2+c^2+d^2=0$ (Eq.(\ref{SymL-x2-a2-ax})) then becomes 
$(1+\gamma^2)(a^2+d^2)=0$, implying $a^2+d^2=0$. On the contrary, from  $a^2+b^2=x^2+y^2$ (Eq.(\ref{213-c2+d2-z2-w2})), we obtain  $(1+\gamma^2)(a^2-d^2)=0$, 
which implies $a^2-d^2=0$. These two equations lead to $a=d=0$, a contradiction. Thus, we obtain $\gamma^2=-1$. For $\sigma_3$, 
$f_{\sigma_3}$ can be expressed as $[-\gamma,1](x_3)\cdot (\gamma a,e_3d,d,\gamma e_1a,\gamma a,e_3d,d,\gamma e_1a)$, which implies $f\in\DUP$, a contradiction. 

\s
{[Subcase: $bd=0$]}
Firstly, we assume that $b=d=0$. In this case,  
$f_{\sigma_3}$ equals $(a,x,c,z,0,0,0,0)$, a contradiction against $f\not\in\DUP$.  
Secondly, we assume that $b=0\wedge d\neq0$. {}From  
$a^2+b^2+c^2+d^2=0$ (Eq.(\ref{SymL-x2-a2-ax})) and  $a^2+b^2=x^2+y^2$ (Eq.(\ref{213-c2+d2-z2-w2})), we obtain $a^2+c^2+d^2=0$ and $a^2-c^2-d^2=0$, respectively. Combining these two equations leads to $2a^2=0$. This is a contradiction against $a\neq0$.
Finally, we assume that $b\neq0\wedge d=0$. Applying (1'), we then obtain  $(1-e_1e_2)ab=0$, which yields $e_2=e_1$. Similar to the second case, from  
$a^2+b^2+c^2+d^2=0$ (Eq.(\ref{SymL-x2-a2-ax})) and  
$a^2+b^2=x^2+y^2$ (Eq.(\ref{213-c2+d2-z2-w2})), we conclude that $c=0$. Hence, $a^2+b^2+c^2+d^2=0$ becomes $a^2+b^2=0$. Now, we set $b=\gamma a$ with an appropriate constant $\gamma\in\{\pm i\}$. With this $\gamma$,  $f_{\sigma_3}$ is written as $a\cdot[1,\gamma](x_3)\cdot (1,0,0,e_1,1,0,0,e_1)$, which clearly belongs to $\DUP$, a contradiction.  

\section{Proof of Proposition \ref{prop-type-III}}\label{sec:type-III}

This last section will prove Proposition \ref{prop-type-III},  completing the whole proof of Proposition \ref{h0+h2=h1+h3=0}. As we have done in Sections \ref{sec:proof-SymL}--\ref{sec:type-II}, we set $f=(a,b,c,d,x,y,z,w)$ and let  $SymL(f_{\sigma})=[g_0^{\sigma},g_1^{\sigma},g_2^{\sigma}]$ for each permutation  $\sigma\in S_3$.  

In this proof, we assume that $f$ is $SIG_1$-legal; namely,  $Sym(f)=[h_0,h_1,h_2,h_3]$ satisfies that $h_0+h_2 =h_1+h_3 =0$ and $h_0\neq \xi h_1$ for any value $\xi\in\{\pm i\}$. 
Moreover, we assume that $g_{0,\varepsilon}^{\sigma}+g_{2,\varepsilon}^{\sigma}=0$ holds for every permutation $\sigma\in S_3$ and for almost all values of $\varepsilon$. 
Since the degree of this  polynomial equation is at most two, in the rest of this proof, we fix an appropriate value $\varepsilon$ and assume that  $g_{0,\varepsilon}^{\sigma}+g_{2,\varepsilon}^{\sigma}=0$ for every 
$\sigma\in S_3$. 
For simplicity, hereafter, we omit subscript ``$\varepsilon$.'' 
To proceed our proof by contradiction, we further assume that $f\not\in\DUP$.  Notice that, as discussed in Section \ref{sec:SymL-sit-2}, Eq.(\ref{SymL-x2-a2-ax})--(\ref{azcx-bwdy}) should be satisfied.

First, we fix $\sigma=(x_1x_2x_3)$ and, for this $\sigma$, we want to prove  that  $(a+d)(y+z)(x+w)\neq0$ and $xw=yz$. Let us begin with the poof of $(a+d)(y+z)(x+w)\neq0$.   

\begin{claim}\label{from-a+d-to-x+w}
$(a+d)(y+z)(x+w)\neq0$.
\end{claim}

\begin{proof}
Our proof goes by way of contradiction: namely, assuming $(a+d)(y+z)(x+w)=0$, we aim at drawing a contradiction. This assumption implies that at least one of the following three terms must be zero: $a+d$, $y+z$, and $x+w$. In what follows, we consider the situation in which $a+d=0$ is satisfied. The other two possible situations can be treated similarly. 
It follows from $(a+d)(b+c)+(x+w)(y+z)=0$ (Eq.(\ref{xwyz-adbc})) that  (1')  $(x+w)(y+z)=0$; thus, either $x+w=0$ or $y+z=0$ should hold.  

\s
{\bf [Case: $x+w=0$]} 
Note that $w=-x$. {}From  $ax+by+cz+dw=0$ (Eq.(\ref{SymL-x2-a2-ax})), we obtain (2') $2ax+by+cz=0$. Moreover, the equation   $ac+bd+xz+yw=0$ (Eq.(\ref{SymL-a2-c2-ac})) implies  (3') $a(b-c)+x(y-z)=0$. Hereafter, we will examine four subcases, depending on the values of $a$ and $x$. 

\s
[Subcase: $ax\neq0$]  Let $\gamma =\frac{x}{a}$.  Note that $\gamma\neq0$. {}From (3'), 
we obtain both (4') $x=\gamma a$ and (5') $b-c=-\gamma(y-z)$.  Next, we use   $x^2+y^2+z^2+w^2=0$ (Eq.(\ref{SymL-x2-a2-ax})) and then obtain  (6') $2\gamma^2a^2+y^2+z^2=0$. 
Since  $b^2-c^2+y^2-z^2=0$ (Eq.(\ref{a2+b2-relation})) is equivalent to $(b+c)(b-c)+(y+z)(y-z)=0$,  (5') implies  (7') $(y-z)[(y+z)-\gamma(b+c)]=0$. 

(i) First, assume that $y=z$; thus,  $b=c$ also holds by (5'). We can deduce  (8') $y^2+\gamma^2a^2=0$ from (6'). In addition, applying (2'), we obtain (9') $\gamma a^2+by=0$. Now, we calculate 
(8') $-$ (9')$\times \gamma$. We then obtain  $y^2-\gamma by=0$, or equivalently  
$y(y-\gamma b)=0$. This equation gives  $y=\gamma b$, and hence $f$ becomes  $(a,b,b,-a,\gamma a,\gamma b,\gamma b, -\gamma a)$, which is also written as  $[1,\gamma](x_1)\cdot (a,b,b,-a,a,b,b,-a)$. Obviously, $f$ belongs to 
 $\DUP$, a contradiction. 

(ii) On the contrary, we assume that $y\neq z$. This inequality implies (10') $y+z=\gamma(b+c)$ by (7').  By calculating 
(10') $+$ (5')$\times\gamma$, we obtain  (11')  $2\gamma b = (1-\gamma^2)y+(1+\gamma^2)z$. Similarly, by calculating (10') $-$ (5')$\times\gamma$, we easily obtain  (12') $2\gamma c=(1+\gamma^2)+(1-\gamma^2)z$. 
It then follows from $ax+by+cz+dw=0$ (Eq.(\ref{SymL-x2-a2-ax})) that  $2\gamma a^2+by+cz=0$; thus, (13') $2\gamma(by+cz+2\gamma a^2)=0$ holds. 
By inserting  (11')\&(12') and $2\gamma^2a^2=-(y^2+z^2)$ obtained from (6') into (13'), we deduce the equation  $(1-\gamma^2)(y^2+z^2)+2(1+\gamma^2)yz-2(y^2+z^2)=0$, which is simplified as $(1+\gamma^2)(y-z)^2=0$. Since $y\neq z$, we conclude that  $\gamma^2=-1$. Using this value, we can draw from (11')\&(12') the consequences: $y=\gamma b$ and $z=\gamma c$. Hence, $f$ is of the form $(a,b,c,-a,\gamma a,\gamma b,\gamma c,-\gamma a)$. This makes $f$ fall into $\DUP$, a contradiction. 

\s
[Subcase: $a=x=0$] 
{}From  the equation $a^2+b^2+c^2+d^2=0$ (Eq.(\ref{SymL-x2-a2-ax})), it follows that  $b^2+c^2=0$. Similarly,  $x^2+y^2+z^2+w^2=0$ (Eq.(\ref{SymL-x2-a2-ax})) implies  $y^2+z^2=0$. Inserting these equations into  $b^2-c^2+y^2-z^2=0$ (Eq.(\ref{a2+b2-relation})), we obtain $b^2+y^2=0$. Now, let $y=\gamma b$ using an appropriate constant $\gamma\in\{\pm i\}$. It then follows from   $y^2+z^2=0$ that (1') $\gamma^2b^2+z^2=0$. In addition,   $ax+by+cz+dw=0$ (Eq.(\ref{SymL-x2-a2-ax})) leads to  (2') $\gamma b^2+cz=0$. Next, we calculate   (2')$\times\gamma$ $-$ (1') and then obtain  (14') $z(z-\gamma c)=0$.

Here, we assume that $z=0$. Since this assumption implies $y=b=c=0$,  $f$ becomes an all-zero function, belonging to $\DUP$, a contradiction. 
On the contrary, we assume that $z\neq0$; thus, (14') implies $z=\gamma c$. Obviously, 
$f$ is of the form $(0,b,c,0,0,\gamma b,\gamma c,0)$, which is also 
in $\DUP$. 

\s
[Subcase: $a=0\wedge x\neq0$] 
{}From (3'), we immediately obtain $x(y-z)=0$, yielding  $y=z$. {}From  $b^2-c^2 +y^2-z^2=0$ (Eq.(\ref{a2+b2-relation})), we also obtain (15') $b^2=c^2$. Moreover, from  $a^2+b^2+c^2+d^2=0$ (Eq.(\ref{SymL-x2-a2-ax})) follows (16')  $b^2+c^2=0$. Using (15')--(16'), we deduce  $b=c=0$. Overall, $f$ must have the form $(0,0,0,0,x,y,y,-x)$, indicating that $f\in\DUP$, a contradiction.   

\s
[Subcase: $a\neq0\wedge x=0$]  This subcase is similar to the previous subcase for $a=x=0$ and is omitted.  

\s
{\bf [Case: $x+w\neq0$]}   
Assume that $x+w\neq0$. By (1'),  $x+w\neq0$ implies $y+z=0$. Let us recall the equation $a^2-d^2+x^2-w^2=0$ (Eq.(\ref{a2+b2-relation})), which is equivalent to $(a-d)(a+d)+(x-w)(x+w)=0$. Since $a+d=0$, we obtain $(x-w)(x+w)=0$. By our assumption, it follows that $x=w$; thus, $x$ cannot be zero. Next, we use the equation  $x^2+y^2+z^2+w^2=0$ (Eq.(\ref{SymL-x2-a2-ax})) to obtain (17') $x^2+y^2=0$. Here, we let $y=\delta x$ for a certain constant $\delta\in\{\pm i\}$. 
The equation $ax+by+cz+dw=0$ (Eq.(\ref{SymL-x2-a2-ax})) leads to  $y(b-c)=0$; thus, either $y=0$ or $b=c$ holds. 

We begin studying the case $y=0$. By (17'), we immediately conclude that $x=0$, a contradiction. Next, we consider the case $b=c$. 
The equation  $a^2+b^2+c^2+d^2=0$ (Eq.(\ref{SymL-x2-a2-ax})) then becomes  $a^2+b^2=0$. Now, we set $b=\gamma a$ using an appropriate constant $\gamma\in\{\pm i\}$. There are two subcases to examine. When $\gamma =-\delta$ is satisfied, for the permutation $\sigma_2=(x_2x_1x_3)$, 
$f_{\sigma_2}$ can be expressed as $[1,\gamma](x_2)\cdot (a,\gamma a,x,-\gamma x,a,\gamma a,x,-\gamma x)$, which is obviously in $\DUP$. On the contrary, when $\gamma=\delta$, for $\sigma_3=(x_3x_2x_1)$,  $f_{\sigma_3}$ becomes $[1,\gamma](x_3)\cdot (a,x,\gamma a,-\gamma x,a,x,\gamma a,-\gamma x)$, and thus $f$ falls into $\DUP$.  This contradicts $f\not\in\DUP$. 
\end{proof}

What we need to prove next is the equality  $xw=yz$. Note that, by Claim \ref{from-a+d-to-x+w}, none of the following terms is zero: $a+d$, $y+z$, and $x+w$. We will use this fact in the proof of Claim \ref{xw-equal-yz}. 

\begin{claim}\label{xw-equal-yz}
$xw=yz$.
\end{claim}

\begin{proof}
Since $a+d\neq0$, let $\gamma =\frac{x+w}{a+d}$; thus, we obtain two equations: (1') $x+w=\gamma(a+d)$ and (2') $b+c=-\gamma(y+z)$. 
Note that $b^2-c^2+y^2-z^2=0$ (Eq.(\ref{a2+b2-relation})) is equivalent to $(b-c)(b+c)+(y-z)(y+z)=0$. We insert (2') to this equation and then obtain $(y+z)[(y-z)-\gamma(b-c)]=0$. Moreover, since $y+z\neq0$, it follows that (3') $y-z=\gamma(b-c)$. 
To remove the term $c$, we calculate (2')$\times\gamma$ $+$ (3') and then  obtain (10') $2\gamma b= (1-\gamma^2)y-(1+\gamma^2)z$. Similarly, by calculating  (2')$\times \gamma$ $-$ (3'), we obtain (11')  $2\gamma c= -(1+\gamma^2)y +(1-\gamma^2)z$. These equations help evaluate the term $2\gamma(by+cz)$ as  $2\gamma(by+cz)= (1-\gamma^2)(y^2+z^2) - 2(1+\gamma^2)yx$, which is obviously equivalent to  (7')  $2\gamma(by+cz) = (1+\gamma^2)(y-z)^2-2\gamma^2(y^2+z^2)$. 

In a similar manner, since  $a^2-d^2+x^2-w^2=0$ (Eq.(\ref{b2+d2-relation}))  is equivalent to $(a-d)(a+d)+(x-w)(x+w)=0$, we insert (1') and then obtain  $(a+d)[(a-d)+\gamma(x-w)]=0$, implying  (4') $a-d =-\gamma(x-w)$.  By calculating (1') $+$ (4')$\times\gamma$, we obtain  (5') $2\gamma a=(1-\gamma^2)x+(1+\gamma^2)w$. Similarly, the calculation of  (1') $-$ (4')$\times \gamma$ shows  (6') $2\gamma d = (1+\gamma^2)x+(1-\gamma^2)w$. 
This implies (8') $2\gamma (ax+ dw) = (1+\gamma^2)(x+w)^2 - 2\gamma^2(x^2+w^2)$. 

Inserting (7')--(8') into  $2\gamma(ax+by+cz+dw)=0$ (Eq.(\ref{SymL-x2-a2-ax})), we obtain $(1+\gamma^2)[(x+w)^2+(y-z)^2] -2\gamma(x^2+y^2+z^2+w^2)=0$. Since  $x^2+y^2+z^2+w^2=0$ (Eq.(\ref{SymL-x2-a2-ax})), it holds that (9')  $(1+\gamma^2)[(x+w)^2+(y-z)^2]=0$.  Now, we examine two possible cases. 

(i) First, assume that  $\gamma^2=-1$. By (5')--(6') and (10')--(11'), 
it follows that $2\gamma b = 2y$, $2\gamma c= 2z$, $2\gamma a=2x$, and $2\gamma d=2w$; in other words, $y=\gamma b$, $z=\gamma c$, $x=\gamma a$, and $w=\gamma d$. These values make  $f$ equal $(a,b,c,d,\gamma a,\gamma b,\gamma c,\gamma d)$, which can be written as $[1,\gamma](x_1)\cdot (a,b,c,d,a,b,c,d)$. Hence, $f$ clearly belongs to $\DUP$, a contradiction. 

(ii) Assume that $\gamma^2\neq-1$; thus, (9') implies  $(x+w)^2+(y-z)^2=0$, which is the same as $x^2+y^2+z^2+w^2+2(xw-yz)=0$. Since   $x^2+y^2+z^2+w^2=0$ (Eq.(\ref{SymL-x2-a2-ax})), we conclude that  $xw=yz$.
\end{proof}

By this point, we have proven, for $\sigma=(x_1x_2x_3)$, that both  $(a+d)(y+z)(x+w)\neq0$ and $xw=yz$ hold. By simply permuting the variable indices,  a similar argument can show that, for $\sigma_2=(x_2x_1x_3)$, both  $(a+y)(d+z)(c+w)\neq0$ and $cw=dz$ hold. Similarly, when $\sigma_3=(x_3x_2x_1)$, we obtain  both $(a+z)(d+y)(b+w)\neq0$ and $bw=dy$.  To complete the proof of Proposition \ref{prop-type-III}, we consider four cases separately. 

\s
{\bf [Case: $xy\neq0$]}  Now, let $\delta=\frac{y}{x}$. This implies that $y=\delta x$ and $w=\delta z$. The assumption $y\neq 0$ implies that $\delta\neq0$. {}From $cw=dz$, we obtain  $\delta cz=dz$, implying  $z(d-\delta c)=0$. Hence, $d=\delta c$ follows. Using  $b^2+d^2+y^2+w^2=0$ (Eq.(\ref{SymL-a2-b2-ab})), we obtain  $b^2+\delta^2(c^2+x^2+z^2)=0$. Applying $a^2 = -(c^2+x^2+z^2)$, which is obtained from 
$a^2+c^2+x^2+z^2=0$  (Eq.(\ref{SymL-a2-b2-ab})), we conclude that    $b^2-\delta^2a^2=0$; thus, either $b=\delta a$ or $b=-\delta a$ holds. 
First, let us consider the case where  $b=-\delta a$. It follows from   $ab+cd+xy+zw=0$ (Eq.(\ref{SymL-a2-b2-ab})) that $-\delta a^2+\delta(c^2+x^2+z^2)=0$. As discussed before, this is equivalent to $-\delta a^2+\delta(-a^2)=0$, which yields $-2\delta a^2=0$. Since  $a^2=0$, we obtain $b=0$. This implies that, for the permutation $\sigma_3=(x_3x_2x_1)$,  $f_{\sigma_3} =(0,x,c,z,0,\delta x,\delta c,\delta z)$; thus, $f$ is in $\DUP$,  a contradiction.
For the next case where $b=\delta a$, $f_{\sigma_3}$ also equals $(a,x,c,z,\delta a,\delta x,\delta c,\delta z)$ and $f$ thus falls into $\DUP$, a contradiction.  

\s
{\bf [Case: $x=y=0$]} 
Note that, since $x+w\neq0$, $x=0$ implies $w\neq0$. Since  $(y+z)(a+y)(d+y)\neq0$, $y=0$ implies $zad\neq0$. Moreover, from 
$bw=dy$, we obtain $bw=0$; thus,  $b=0$ follows. 
{}From  $x^2+y^2+z^2+w^2=0$ (Eq.(\ref{SymL-x2-a2-ax})), we obtain  $z^2+w^2=0$. Here, let $z=\gamma w$ for a certain constant $\gamma\in\{\pm i\}$.  It then follows from  $ax+by+cz+dw=0$ (Eq.(\ref{SymL-x2-a2-ax})) that $\gamma cw+dw=0$, implying $w(d+\gamma c)=0$. Hence, we obtain  $d=-\gamma c$. Finally, $a^2+b^2+c^2+d^2=0$ (Eq.(\ref{SymL-x2-a2-ax})) implies $a^2=0$. This proves that $a+y=0$, a contradiction.

\s
{\bf [Case: $x=0\wedge y\neq0$]} 
Since  $x+w\neq0$, it holds that  $w\neq0$. Let $\gamma=\frac{w}{y}$. By $bw=dy$, we obtain  $w=\gamma y$ and $d=\gamma b$. Moreover, from $xw=yz$ follows $z=0$. We thus obtain  $c=0$ from $cw=dz$.  
It then follows from   $x^2+y^2+z^2+w^2=0$ (Eq.(\ref{SymL-x2-a2-ax}))    that $(1+\gamma^2)y^2=0$; thus, $\gamma^2=-1$. Here, the equation  $a^2+b^2+c^2+d^2=0$ (Eq.(\ref{SymL-x2-a2-ax})) implies  $a^2+(1+\gamma^2)b^2=0$, which immediately yields $a^2=0$. Hence, $f_{\sigma_3}$ is of the form $(0,b,0,y,0,\gamma b,0,\gamma y)$, making $f$ fall into $\DUP$, a contradiction. 

\s
{\bf [Case: $x\neq0\wedge y=0$]} 
Since $(y+z)(a+y)(d+y)\neq0$, $y=0$ implies $zad\neq0$.  The equation  $xw=yz$ leads to $xw=0$, implying $w=0$. Moreover, $cw=dz$ implies $dz=0$. This contradicts the result $zad\neq0$.  

\s

In this end, we have completed the proof of Proposition \ref{prop-type-III}.

\bibliographystyle{alpha}

\end{document}